\pdfoutput=1
\documentclass[12pt]{article}

\usepackage[margin=2.5cm]{geometry}

\usepackage{amsmath,amsthm,amsfonts,amscd,amssymb,bbm,mathrsfs,enumerate,url}

\usepackage{graphicx,tikz}
\usepackage[pdfborder={0 0 0}]{hyperref}
\usetikzlibrary{matrix,arrows}
\usetikzlibrary{decorations.pathreplacing}
\usetikzlibrary{decorations.pathmorphing}
\usetikzlibrary{patterns,fadings}

\newcounter{axioms}

\def\RR{{\mathbb R}}
\def\CC{{\mathbb C}}
\def\NN{{\mathbb N}}
\def\ZZ{{\mathbb Z}}

\def\A{{\mathcal A}}
\def\B{{\mathcal B}}

\def\F{{\mathcal F}}
\def\H{{\mathcal H}}

\def\K{{\mathcal K}}
\def\M{{\mathcal M}}
\def\N{{\mathcal N}}

\def\P{{\mathcal P}}
\def\R{{\mathrm R}}

\def\U{{\mathcal U}}
\def\V{{\mathcal V}}
\def\W{{\mathcal W}}

\def\Poi{{\mathcal {P}_+^\uparrow}}

\def\a{\alpha}

\def\G{\Gamma}

\def\i{\iota}
\def\k{\kappa}

\def\l{\lambda}

\def\L{{\mathrm L}}

\def\Ad{{\hbox{\rm Ad\,}}}

\def\sp{{\rm sp}\,}

\def\1{{\mathbbm 1}}

\def\dual{{\rm d}}

\def\diff{{\rm Diff}}

\def\diffs1{\diff(S^1)}
\def\mob{{\rm M\ddot{o}b}}
\def\uMob{{\widetilde{\Mob}}}
\def\vir{{\rm Vir}}

\def\supp{{\rm supp\,}}

\def\psl2r{{\rm PSL}(2,\RR)}
\def\sl2r{{\rm SL}(2,\RR)}
\def\su11{{\rm SU}(1,1)}
\def\2dmob{{\overline{\psl2r}\times\overline{\psl2r}}}
\def\<{\langle}
\def\>{\rangle}

\def\Im{\mathrm{Im}\,}
\def\re{\mathrm{Re}\,}
\def\im{\mathrm{Im}\,}

\def\poincare{{\P^\uparrow_+}}
\def\poincarej{{\P_+}}

\def\cyl{{\widetilde{M}}}

\DeclareMathOperator{\Tr}{Tr}

\DeclareMathOperator{\Mob}{M\ddot ob}

\newtheorem{theorem}{Theorem}[section]

\newtheorem{corollary}[theorem]{Corollary}
\newtheorem{proposition}[theorem]{Proposition}
\newtheorem{lemma}[theorem]{Lemma}
\theoremstyle{remark}
\newtheorem{remark}[theorem]{Remark}

\title{Scale and M\"obius covariance in two-dimensional Haag-Kastler net}
\date{} 
\author{
{\bf Vincenzo Morinelli}\footnote{Supported in part by the ERC Advanced Grant 669240 QUEST ``Quantum Algebraic Structures and
Models'', MIUR FARE R16X5RB55W QUEST-NET, GNAMPA-INdAM .}, \quad
{\bf Yoh Tanimoto}\footnote{Supported by Programma per giovani ricercatori, anno 2014 ``Rita Levi Montalcini''
of the Italian Ministry of Education, University and Research.
}
\\
   Dipartimento di Matematica, Universit\`a di Roma Tor Vergata\\
   Via della Ricerca Scientifica 1, I-00133 Roma, Italy\\
   email: {\tt morinell@mat.uniroma2.it}, {\tt hoyt@mat.uniroma2.it}
      \vspace{0.5cm}
}

\begin{document}
\maketitle

\begin{abstract}
Given a two-dimensional Haag-Kastler net which is Poincar\'e-dilation covariant
with additional properties, we prove that it can be extended to a M\"obius covariant net.
Additional properties are either a certain condition on modular covariance,
or a variant of strong additivity. The proof relies neither on the existence of stress-energy tensor
nor any assumption on scaling dimensions.
We exhibit some examples of Poincar\'e-dilation covariant net which cannot be extended to
a M\"obius covariant net, and discuss the obstructions.
\end{abstract}

\section{Introduction}\label{introduction}
For a relativistic quantum field theory, there has been a long-standing question
whether scale invariance (dilation covariance) implies conformal covariance \cite{Nakayama15}.
In $(1+1)$-dimensions, we call the latter {\it M\"obius covariance}\footnote{The word ``conformal'' is reserved for diffeomorphism covariance,
c.f.\! \cite{KL04-2}.}
in order to distinguish it from diffeomorphism covariance (an action of the Virasoro algebra).
This claim, of course, should not be taken literally.
A simple counterexample can be given based on a generalized free field
which is dilation covariant but not M\"obius or conformally covariant
(see Section \ref{counterexample}).
On the other hand, in $(3+1)$-dimensions, there is no known example (even in the physical sense)
of relativistic (unitary) dilation-covariant quantum field theory with certain additional conditions
which is not conformally covariant,
although there is currently no proof of the enhancement either.
In $(1+1)$-dimensions, the implication ``dilation $\Longrightarrow$ M\"obius''
is considered a ``theorem'', whose proof exploits the existence
of stress-energy tensor and the discreteness of scaling dimension \cite{Zomolodchikov86, Polchinski88}.

Dilation covariance is believed to appear naturally in physical models.
If one looks at longer and longer length scale of a physical system,
the behaviour of the system should not depend on the details in the smaller scale
and may obtain a low-energy effective theory which is scale invariant.
Alternatively, one might look at smaller and smaller spacetime regions in the quantum chromodynamics,
and should be able to see quarks which are otherwise confined and not visible.
Such a limiting theory is expected to be simpler and to obtain the dilation symmetry
(yet this is not automatic, see \cite{BDM10}).
Now, a dilation-covariant theory has often an additional symmetry, the conformal symmetry.
Indeed, in $(1+1)$-dimensions, most of important dilation-covariant theories are indeed
M\"obius covariant. Although not all dilation-covariant theories have
M\"obius covariance, it is natural to expect some additional conditions
should imply the enhancement of symmetry.

In theoretical physics, 
the problem is considered to be solved in $(1+1)$-dimensions by the argument by
Zamolodchikov \cite{Zomolodchikov86} 
and Polchinski \cite{Polchinski88}, which are based on the existence of scale current
and the discreteness of scale dimensions.
On the other hand, the enhancement of symmetry can be clearly stated even in terms of axiomatic/algebraic
quantum field theory, hence it is natural to expect that certain additional assumptions
should really imply M\"obius covariance in the mathematical level.
In this respect, 
Guido, Longo and Wiesbrock proved that a dilation-translation covariant net of von Neumann algebras
on the real line $\RR$ satisfying the Bisognano-Wichmann property can be extended to the compactified
real line $S^1$ and obtain  M\"obius covariance \cite{GLW98}.
Remarkably, this last result does not assume any other physical requirement such as
the existence of current or scaling dimensions,
but the proof is based on the modular theory of von Neumann algebras.
Hence one might expect a similar result for two-dimensional dilation-covariant quantum field theories.

In this paper, we present a proof of enhancement to M\"obius covariance, in addition to the standard Haag-Kastler axioms,
under the following operator-algebraic conditions
\begin{itemize}
 \item The vacuum is cyclic and separating for the lightcone algebra.
 \item The theory is covariant under dilations and it is implemented by the modular group for the lightcone algebra,
 and one of the following holds.
 \begin{enumerate}[{(}a{)}]
  \item The half-hand algebra and the double cone algebra (see Figure \ref{fig:regions})
  consist a half-sided modular inclusion \cite{Wiesbrock93-1, AZ05}
  \item The theory satisfies a variation of strong additivity.
 \end{enumerate}
\end{itemize}
Only one of the last two conditions is needed in our proof.
The first one is an assumption about the modular group of a certain infinitely extended region, and might look too strong,
but actually we show that it is a consequence of the second, which appears to have little to do with conformal covariance.
With these conditions, which again are not concerned with either the current/stress-energy tensor or
scaling dimensions, we can extend the symmetry group to the two-dimensional M\"obius group
by the modular theory.

We present two families of counterexamples.
In one of them, we simply break the Bisognano-Wichmann property for the future lightcone $V_+$ which is a necessary condition for M\"obius covariance \cite{GLW98}.
In the other, we take a certain representation of the two-dimensional
M\"obius group and apply the BGL construction \cite{BGL93}. This itself is M\"obius covariant, but its dual net
is the dual net of a generalized free field which cannot be M\"obius covariant. This last example provides
a M\"obius covariant net with the trace class property whose dual net is not M\"obius covariant and does not have the split property.
The reason why this dual net cannot be M\"obius covariant (the vacuum is not separating for the algebra of $V_+$) is
different from the reason why some generalized free fields cannot be M\"obius covariant (wrong scaling dimension).
We also examine arguments in physics literature and see to which extent they work.

This paper is organized as follows.
In Section \ref{preliminaries} we explain the geometric setting and the symmetry structure of two-dimensional M\"obius covariant net,
and state our additional assumptions on dilation covariant nets.
In Section \ref{proofmc} we give a proof of M\"obius covariance based on these assumptions.
In Section \ref{counterexample}, examples of dilation-covariant nets which do not extend to M\"obius covariant nets are provided.
In Section \ref{comments} we discuss to what extent our assumptions are necessary,
some arguments in physics literature and open problems.
Beside, we need the two-dimensional spin statistics theorem in the course of the proof,
and we exhibit a proof in Appendix \ref{spinstatistics} for self-containedness.
In Appendix \ref{directintegral} we provide basic results on direct integrals of Hilbert subspaces which we were not able to find in literature,
and are essential for our counterexamples.

\section{Preliminaries}\label{preliminaries}
Here we are going to describe the operator-algebraic setting for quantum field theory
and various spacetime symmetries.
\subsection{One-dimensional M\"obius group}
The $(1+1)$-dimensional Minkowski space is the product of two lightrays.
The subgroup of lightlike translations and dilations of the Poincar\'e group
acts on each lightray $\overline \RR = \RR \cup \{\infty\}$.
This action can be extended to the M\"obius group $\Mob = \mathrm{PSL}(2,\RR) \cong \mathrm{PSU}(1,1)$,
which we review here.
See \cite{Longo08, Weinerthesis} for our notations.

Consider the Cayley transform $C:\overline\RR\ni x\mapsto -\frac{x-i}{x+i}\in S^1$,
 where $S^1$ is the complex unit circle  $\{z\in\CC: |z|=1\}$ and $C(\infty)$ is defined to be equal to $-1$ by convention.
 With this map $C$, we can pass from the line to the circle picture. The Cayley transform is the inverse of the stereographic projection
 $C^{-1}:S^1\ni z\mapsto -i\,\frac{z-1}{z+1}\in\overline\RR$ and sends $S^1$ onto $\overline\RR$. 
 With this convention, the upper semicircle is mapped in to the right half-line $(0,+\infty)$.

 The group $\mathrm{SL}(2,\RR)$ acts on the compactified line $\overline \RR$
 by linear fractional transformations.
 The kernel of its action is $\{\pm1\}$ and $\mathrm{PSL}(2,\RR)=\mathrm{SL}(2,\RR)\slash\{\pm1\}$ defines the {\bf M\"obius group},
 the group of orientation preserving conformal transformations of $\overline \RR$. We denote it by $\Mob$.
 On the unit circle in $\CC$, the action of $\Mob$ translates to that of $\mathrm{SU}(1,1)$ again by linear fractional transformations
 through the Cayley transform.
 
The group $\Mob$ is a three-dimensional Lie group and can be generated by the following one-parameter subgroups:
 \begin{itemize}
 \item Rotations $\rho_\theta$: for $\theta \in \RR/2\pi\ZZ$, $\rho_\theta = e^{i\theta}z\in S^1$,  in the circle picture;
 \item Dilations $\delta_s$: for $s \in \RR$, $\delta_s a =  e^s a\in \RR$, on the line picture.
 \item Translation $\tau_t$: for $t \in \RR$, $\tau_t a = a +s\in \RR$, on the line picture;
 \end{itemize}
In literature
they are respectively denoted  with $\mathrm{\mathbf{K}}$, $\mathrm{\mathbf{A}}$ and $\mathrm{\mathbf{N}}$,
and any element $g\in\Mob$ can be uniquely decomposed 
following the $\mathrm{\mathbf{KAN}}$ decomposition (Iwasawa decomposition), i.e.\!
the product of elements from each of these groups.
The subgroups $\mathrm{\mathbf{A}}$ and $\mathrm{\mathbf{N}}$
generate the translation-dilation group  $\mathrm{\mathbf{P}}$ which preserves the point $\infty$ in the real line picture.

In general, an element $g \in \Mob$ is determined by its action on three points of the circle.
Any pair of points on $S^1$, hence any interval on $S^1$, can be brought to another pair, respectively another interval,
by a M\"obius transformation. If $g \in \Mob$ takes $\RR_+$ (in the line picture) to a general  interval $I$
(in the circle picture), then we denote by $\Lambda_I(t) = g\delta_{-t} g^{-1}$,
and call them the dilations\footnote{By convention, the sign is reversed: $\Lambda_{\RR_+}(t) = \delta_{-t}$,
in accordance with \cite{GLW98}.} associated with $I$.
Note that $\Lambda_I$ does not depend on the choice of $g$.
By this correspondence, $\Lambda_{\RR_-}(t) = \delta_{t}$,
and $\Lambda_{\RR_+ + 1}(t)\cdot a = e^{-t}(a-1)+1$.
The two subgroups $\{\Lambda_{\RR_+}(t)\}$ and $\{\Lambda_{\RR+1}(t)\}$ generate a two-dimensional subgroup of $\Mob$
which is isomorphic to $\textbf{P}$ and preserves the point $\infty$.
Furthermore, any element of a small neighborhood of the unit element can be written as a simple product:
indeed, we have
\begin{align*}
 \Lambda_{\RR_+}(t)\Lambda_{\RR_+ + 1}(s) = \Lambda_{\RR_++1}\left(-\ln({e^{-t-s}+1-e^{-t}})\right)
 \Lambda_{\RR_+}\left(-\ln\left(\frac{e^{-t-s}}{{e^{-t-s}+1-e^{-t}}}\right)\right)
\end{align*}
and it is immediate that any finite product of $\Lambda_{\RR_+}$ and $\Lambda_{\RR_+ + 1}$,
as long as the parameters are sufficiently small (namely
when $e^{-t-s}+1-e^{-t} > 0$),
can be reduced to a product of two in the desired order.
A similar relation holds for $\Lambda_{\RR_+}$ and $\Lambda_{(0,1)}$:
 \begin{align}\label{eq:commutation}
  \Lambda_{\RR_+}(t)\Lambda_{(0,1)}(s) = 
  \Lambda_{(0,1)}\left(\ln({e^{t+s}+1-e^{t}})\right)
  \Lambda_{\RR_+}\left(\ln\left(\frac{e^{t+s}}{{e^{t+s}+1-e^{t}}}\right)\right)
 \end{align}
By bringing the three points $0,1,\infty$ to another three points, an analogous relation holds
for $\Lambda_{I_1}, \Lambda_{I_2}$, where $I_1 \supset I_2$ and there is one and only one of the endpoints 
shared by $I_1$ and $I_2$.
We call these relations simply the {\bf commutation relations} of $\Lambda_{I_1}, \Lambda_{I_2}$.

The group $\Mob$ can be generated by different subgroups. 
Consider $I_1, \,I_2,\, I_3$, disjoint intervals whose union is dense in $S^1$, then
$\Lambda_{I_k}, k=1,2,3$ generate $\Mob$. This can be seen from the fact that they together can move any ordered three points to any other ordered three points. 
In particular, for $I_1 = (-\infty, 1), I_2 = (0,1), I_3 = (0,\infty)$,
$\Lambda_{I_k}, k \in \ZZ_3$ generate $\Mob$ and any pair
$\Lambda_{I_k}, \Lambda_{I_{k+1}}$ generates a subgroup isomorphic to $\mathrm{\mathbf{P}}$.

\subsection{{\texorpdfstring{$(1+1)$}{(1+1)}-dimensional Minkowski space and Einstein cylinder}}\label{2dminkowski}

Consider the set of coordinates given by the lightrays
$\left(\frac{a_0-a_1}{\sqrt2},\frac{a_0+a_1}{\sqrt 2}\right)$,
where $a_0$ is the time coordinate and $a_1$ is the space coordinate.
The following spacetime regions play important roles in our work.
\begin{itemize}
\item Forward and backward light-cones: $V_+=\RR_+\times \RR_+$ and $V_- = \RR_-\times \RR_-$
\item Right and left standard wedges: $W_\R=\RR_-\times\RR_+ $ and $W_\L=\RR_+\times \RR_-$
\item Right and left half-bands: $B_{\R,(c,d)}^\pm=(a,b)\times \RR_\pm$ and $B_{\L,(a,b)}^{\pm}=\RR_\pm\times (a,b)$
\item Double cones: $D_{(a,b),(c,d)}=(a,b)\times(c,d)$ 
\end{itemize}
Then, we take also some specific regions (see Figure \ref{fig:regions}):
\begin{itemize}
\item $B_{\R}=(0,1)\times \RR_+$ and $B_{\L}=\RR_+\times (0,1)$
\item  $D_0=(0,1)\times(0,1)$ 
\end{itemize}

\begin{figure}[ht]\centering
\begin{tikzpicture}[path fading=north,scale=0.75]
         \draw [thick, ->] (-4,0) --(4,0) node [above left] {$a_1$};
         \draw [thick, ->] (0,-1)--(0,5) node [below right] {$a_0$};
         \draw [ultra thick] (0,0)-- (-3.5,3.5) node [above right] {$B_\L$};
         \draw [ultra thick,dotted] (-3.5,3.5)--(-4,4);
         \draw [ultra thick] (1,1) -- (-2.5,4.5);
         \node at(1.2,0.5) {$D_0$};
         \draw [ultra thick] (-1,1)-- (2.5,4.5);
         \draw [ultra thick,dotted] (2.5,4.5)--(3,5);
         \draw [ultra thick,dotted] (-2.5,4.5)--(-3,5);
         \draw [ultra thick,dotted] (3.5,3.5)--(4,4);
         \draw [ultra thick] (3.5,3.5)node [above left] {$B_\R$} -- (0,0);
              \fill [color=black,opacity=0.4]
               (0,0)--(4,4)--(3,5)--(-1,1)--(0,0);
              \fill [color=black,opacity=0.4]
               (0,0)--(-4,4)--(-3,5)--(1,1)--(0,0);
              \fill [color=black,opacity=0.4]
               (0,0)--(1,1)--(0,2)--(-1,1)--(0,0);
\end{tikzpicture}
\caption{Double cone $D_0$ and half-bands $B_\R, B_\L$.}
\label{fig:regions}
\end{figure}
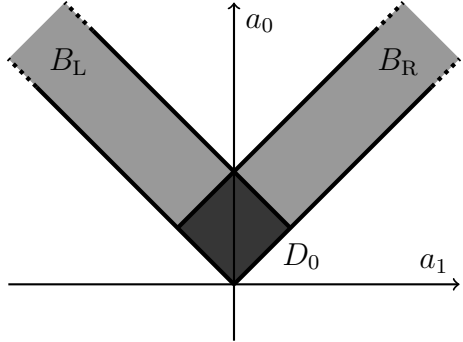

Let $\uMob$ be the universal covering group of $\Mob$.
It is again generated by three one-parameter subgroups $\rho, \delta, \tau$
(we use the same symbols for elements in $\uMob$, as long as non confusion arises),
and $\rho$ is now lifted from $\RR/2\pi\ZZ$ to $\RR$.
Let $G$ be the quotient group of $\uMob\times\uMob$ by the normal subgroup generated by
$(\rho_{-2\pi}, \rho_{2\pi})$. The group $G$ acts locally on the Minkowski space $M$
identified with the product of two lightrays $\RR\times\RR$,
and its action can be promoted to an action on the Einstein cylinder $\cyl$
\cite{BGL93}\footnote{$\cyl$ is homeomorphic to $\RR \times S^1$, but this product decomposition is different from
the lightlike decomposition above: $\RR$ goes in the $a_0$-direction while $S^1$ is the $a_1$-direction.
This decomposition will not be used in this paper.
}.
The Minkowski space is identified with a maximal square $(-\pi,\pi)\times (-\pi,\pi)$,
where the product is intended for the lightlike decomposition parametrized by the lifted rotations.
Let $\i$ be the unit element of $\uMob$.
For any $g\in\uMob$, elements of the form $g\times \i$ (respectively $\i\times g$)
act trivially on the positive lightray $a_0 = a_1$ (respectively the negative lightray $a_0 = -a_1$)

We introduce the following elements in $G$ in terms of the lightlike components:
\begin{itemize}
\item let $\Lambda_{V_+}$ be the two-dimensional dilation of $M$: $\Lambda_{V_+}(t)=\delta(-t)\times\delta(-t)$.
\item let $\Lambda_{W_L}$ be the one-parameter group of Lorentz boosts associated with the left standard wedge $W_L$:
$\Lambda_{W_L}(t)=\delta(-t)\times\delta(t)$.

\end{itemize}
In some literature a different convention is used where the parameter is reversed.
The sign of our convention coincides with that of the modular group (see the Bisognano-Wichmann property in Section \ref{haag-kastler}).

\subsection{The modular theory of von Neumann algebras and half-sided modular inclusions}\label{modular}
Let $\M\subset\B(\H)$ be a von Neumann algebra with a cyclic and separating vector $\Omega\in\H$.
The associated Tomita operator $S_{\M,\Omega}$ is an antilinear involution which is the closure of
\[
\H\supset \M\Omega\ni x\Omega\longmapsto x^*\Omega\in \M\Omega\subset\H. 
\]
Through its polar decomposition $S_{\M,\Omega}=J_{\M,\Omega}\Delta^{\frac12}_{\M,\Omega}$
one obtains the modular conjugation $J_{\M,\Omega}$ and the modular operator $\Delta_{\M,\Omega}$.
They satisfy the relation $J_{\M,\Omega}\Delta_{\M,\Omega} J_{\M,\Omega}=\Delta^{-1}_{\M,\Omega}.$
Furthermore $\Delta_{\M,\Omega}$ is the generator of a one parameter group of automorphisms called
the \textbf{modular automorphism group}, namely $\Delta^{it}_{\M,\Omega}\M\Delta^{-it}_{\M,\Omega}=\M$
(see e.g.\! \cite{TakesakiII}).
For the anti-unitary conjugation $J_{\M,\Omega}$ we have $J_{\M,\Omega}\M J_{\M,\Omega}=\M'$. 
\begin{lemma}\label{lem:fix} 
Let $\M\subset\B(\H)$ be a von Neumann algebra with a cyclic and separating vector $\Omega\in\H$
and a $U$ such that $U\M U^* = \M$ and $U\Omega = \Omega$.
Then it holds that $US_{\M,\Omega}U^*=S_{\M,\Omega}$ and
hence
\[
 U\Delta_{\M,\Omega}U^*=\Delta_{\M,\Omega}, \qquad UJ_{\M,\Omega}U^*=J_{\M,\Omega}.
\]
\end{lemma}
The following theorem, due to Borchers \cite{Borchers92} (and significantly simplified by Florig \cite{Florig98}),
ensures that when there is a one-parameter semigroup of endomorphisms implemented by unitaries with positive generator,
these unitaries and the modular group generate a representation of the group $\mathrm{\mathbf{P}}$
of dilations and translations. In this representation, they are assigned
$\Delta_{\M,\Omega}^{\frac{it}{2\pi}}$ and $U(1)\Delta_{\M,\Omega}^{\frac{is}{2\pi}}U(1)^*$,
respectively and indeed satisfy Equation \eqref{eq:commutation}.

\begin{theorem}[Borchers]\label{thm:bor}
Let $\M\subset\B(\H)$ be a von Neumann algebra with a cyclic and separating vector $\Omega\in\H$,
and $t\mapsto U(t)=e^{iHt}$ be a unitary one parameter group such that $\sp H \subset \RR_\pm$, $U(t)\Omega = \Omega$ and
$\Ad U(t)(\M) \subset \M,\;t\geq0$. Then the following hold:
\begin{align*}
\Delta^{is}_{\M,\Omega}U(t)\Delta^{-is}_{\M,\Omega}&=U(e^{\mp 2\pi s}t),\\
J_{\M,\Omega}U(t)J_{\M,\Omega}&=U(-t),\quad t,s\in\RR .
\end{align*}
\end{theorem} 
The first equality shows that $\Delta_{\M,\Omega}^{is}$ and $U(t)$ provide
a positive energy representation of $\mathrm{\mathbf{P}}$.
We note that the group $\mathrm{\mathbf{P}}$ can be also generated by the $\Lambda_{(0,\infty)}$ and $\Lambda_{(1,\infty)}$,
namely dilations based on $0$ and $1$, respectively.

Let $\N\subset \M\subset\B(\H)$ be an inclusion of von Neumann algebra
with a common cyclic and separating vector $\Omega\in\H$, then the inclusion is said
to be a {\bf half-sided modular inclusion ($\pm$-HSMI)} if
\[
 \Ad \Delta^{-it}_\M (\N)\subset \N,\qquad \pm t \geq 0. 
\]

The following is a fundamental result on HSMIs \cite{Wiesbrock93-1, AZ05}.
\begin{lemma}[Wiesbrock, Araki-Zsido]\label{lm:hsmi}
 If $(\N\subset \M, \Omega)$ is a $+$- (respectively $-$-)HSMI,
 then the modular groups $\Delta_\M^{it}, \Delta_\N^{is}$ satisfy the same commutation relations
 as those of $\Lambda_{\RR_-}$ and $\Lambda_{\RR_- - 1}$ (respectively those of $\Lambda_{\RR_+}$ and $\Lambda_{\RR_+ + 1}$).
\end{lemma}

The following is a slight variation of \cite[Lemma 1.1]{GLW98}, see also \cite[Theorem 6]{Wiesbrock98}.
\begin{lemma}\label{lm:group}
Let $\Upsilon$ be the universal group algebraically generated by 3 one-parameter subgroups
$t_k\mapsto \Lambda_k(t_k), k \in \ZZ_3$, such that $\Lambda_k$ and $\Lambda_{k+1}$ satisfy
the same commutation relation as $\Lambda_{I_k}, \Lambda_{I_{k+1}}$,
where $I_1 = (-\infty, 1), I_2 = (0,1), I_3 = (1, \infty)$ for $t_k$ in an open neighborhood of the origin.
Then $\Upsilon$ can be made a topological group and
there is a continuous isomorphism between $\Upsilon$ and $\uMob$ which intertwines $\Lambda_{I_k}$ and $\Lambda_k$.
\end{lemma}
\begin{proof}
This is essentially covered by \cite[Lemma 1.1]{GLW98}
by noting that the group generated by $\Lambda_{I_k}$ is isomorphic to $ \uMob$.
Indeed, $\Lambda_{\RR_-}=\Lambda_{I_3}^{-1}$, $\Lambda_{(1,\infty)}=\Lambda_{I_1}^{-1}$
and since $(-\infty,0)$, $(0,1)$, $(1,\infty)$ is a factorization of $S^1$,
hence they satisfy the commutation relations of the corresponding intervals.

Let us just make the topology on $\Upsilon$ more explicit.
By the group structure, there is a quotient of $\Upsilon$ which is algebraically isomorphic to $\uMob$.
Let $p$ the quotient map.
We declare that any inverse image of an open set by $p$ is an open set of $\Upsilon$,
and that the topology of $\Upsilon$ is generated by them.
It is easy to check that they form a neighborhood basis, and with this topology,
$\Upsilon$ is a Hausdorff, path-connected (because they are generated by one-parameter groups), locally simply connected space.
It is also immediate that $\Upsilon$ is a topological group,
by taking small neighborhoods of elements in $\Upsilon$ corresponding to
those in $\uMob$.

As $\Upsilon$ is connected and $p$ is a covering map,
$\Upsilon$ is continuously isomorphic $\uMob$ by universality of $\uMob$,
see e.g.\! \cite[Theorem 63]{Pontryagin46}. 
\end{proof}

As pointed out in \cite[Lemma 1.1]{GLW98},
$\Upsilon$ has a natural structure as a Lie group:
there exists an open neighborhood of the origin $\N\subset\RR^3$ and the sets $\U$ and $\V$
containing the identity in $\Upsilon$ and $\uMob$ respectively, such that the maps
\[
 \Phi:\N\ni(t_1,t_2,t_3)\mapsto \Lambda_1(t_1)\Lambda_2(t_2)\Lambda_3(t_3)\in\U
\]
and
\[
 \Phi_\uMob:\N\ni(t_1,t_2,t_3)\mapsto \Lambda_{I_1}(t_1)\Lambda_{I_2}(t_2)\Lambda_{I_3}(t_3)\in\V
\]
are 1-1 and surjective. 
Note that the map $\Phi_\uMob$ is a diffeomorphism respecting
the Lie structure, and through $p$, we can also introduce a manifold structure on $\Upsilon$.
The set $\{g\Phi(\N)\}_{g\in \Upsilon}$ provides an atlas for $\Upsilon$.

\subsection{Haag-Kastler nets}\label{haag-kastler}
Let $\K$ be the set of all the double cones in the Minkowski spacetime $\RR^{1+1}$.
A two-dimensional {\bf Haag-Kastler net} $(\A, U, \Omega)$ is
a net of von Neumann algebras $\{\A(D)\}_{D\in \K}$ in $\B(\H)$ on a fixed Hilbert space $\H$,
together with a strongly continuous unitary representation $U$ of the Poincar\'e group $\P_+^\uparrow$ and
the vacuum vector $\Omega$
satisfying the following assumptions (see e.g.\! \cite[Section 2.1]{Tanimoto12-2})
\begin{enumerate}[{(HK}1{)}]
 \item \textbf{Isotony: }if $D_1\subset D_2$, then $\A(D_1)\subset\A(D_2)$. \label{isotony}
 \item \textbf{Locality:} if $D_1$ and $D_2$ are spacelike separated, then $\A(D_1)\subset\A(D_2)'$. \label{locality}
 \item\textbf{Poincar\'e covariance:} \label{poincare} it holds that
 \[
  U(g)\A(D)U(g)^*=\A(gD), \qquad \text{ for } g\in\P_+^\uparrow,\, D\in K.
 \]

 \item\textbf{Positivity of the energy:}
 the joint spectrum of the translation subgroup in $U$ is contained in the closed forward light cone $\overline {V_+}=\{(a_0,a_1)\in\RR^{1+1}: a_0^2-a_1^2\geq0, a_0\geq0\}$. \label{positiveenergy}

 \item \textbf{Vacuum and the Reeh-Schlieder property: }there exists a unique (up to a phase)
 vector $\Omega \in\H$ such that $U(g)\Omega=\Omega$ for $g\in \poincare$ and is cyclic for any local algebra, namely $\overline{\A(D)\Omega}=\H$. \label{vacuum}
 
 \item\textbf{The Bisognano-Wichmann property:} \label{bw} Let $\Lambda_{W_\L}$
 be the boost one-parameter group associated with the wedge $W_\L$
 (see Section \ref{2dminkowski}), and $\A(W_\L)=\left(\bigvee_{D\subset W_\L} \A(W_\L)\right)'' $ then
 \[
  U(\Lambda_{W_\L}(2\pi t))=\Delta_{\A(W_\L),\Omega}^{it}.
 \]
 \setcounter{axioms}{\value{enumi}}
 \end{enumerate}
We included the Reeh-Schlieder property already in the axioms, because
it follows from weak additivity which is traditionally included in the axioms.
The Bisognano-Wichmann property is not automatic in general (see e.g.\! \cite{Morinelli18}),
but is a consequence of M\"obius covariance (see below) \cite{BGL93}, and it is a natural necessary condition
for M\"obius covariance.

Let $(\A,U,\Omega)$ be a net satisfying (HK\ref{isotony})--(HK\ref{bw}).
We can also define algebras for more general open regions $X$ by $\A(X)=(\bigvee_{D\subset X}\A(D))''$ such as,
for instance, wedges, forward and backward lightcones and half-bands.

A Poincar\'e covariant net $(\A,U,\Omega)$ is said to be {\bf M\"obius covariant}
if the representation $U$ extends to the two-dimensional M\"obius group $G$ (see Section \ref{2dminkowski})
which acts covariantly on the extension of the net $\A$ to the cylinder $\cyl$.
In such a case we shall say that $(\A,U,\Omega)$ is a M\"obius covariant
net.

In order to deduce M\"obius covariance, we further introduce the following conditions.
\begin{enumerate}[{(HK}1{)}]
\setcounter{enumi}{\value{axioms}} 
\item \textbf{Dilation covariance:}\label{dilation} we assume that $U$ extends to representation of the group
of Poincar\'e transformations $\poincare$ and the dilation group $\Lambda_{V_+}$,
which still acts covariantly on the net, namely 
$U(g)\A(D)U(g)^*=\A(gD)$ for all $g$ in the Poincar\'e-dilation group.
 \item \label{rsv} \textbf{Reeh-Schlieder property for $V_+$:} $\Omega$ is cyclic and separating for $\A(V_+)$.
\item \textbf{Bisognano-Wichmann property for dilations:}\label{bwcone}
  \[
   U(\Lambda_{V_+}(2\pi t))=\Delta_{\A(V_+),\Omega}^{it}.
  \]
\item\label{hsmip} One of the following conditions holds.
\begin{enumerate}[{(}a{)}]
 \item\label{hsmib} \textbf{Modular covariance\footnote{This does not require that the actions
 $\Ad \Delta_{B_\L}^{it}$ and $\Ad U(\delta(-2\pi t)\times \Lambda_{(0,1)}(2\pi t))$ should coincide,
 indeed we see that $\Delta_{B_\L}^{it} = U(\delta(-2\pi t)\times\Lambda_{(0,1)}(2\pi t))$,
 where $U$ is an extension to $G$.}:}
 \begin{align*}
 \Ad \Delta_{B_\L}^{it}(\A(D_0)) &= \Ad U(\delta(-2\pi t)\times \iota)(\A(D_0)) \\
 &\left(= \Ad U(\Lambda_{\RR_+}(2\pi t)\times \iota)(\A(D_0))\right), 
 \end{align*}
 especially, $\A(D_0) \subset \A(B_\L)$ is a $+$-HSMI.
 \item\label{sa} \textbf{$\cyl$-strong additivity:} 
Let $a<b<c$ and $d>0$ in $\RR$ and $B_{L,(a,c)}$ and $B_{L,(a,b)}+(d,0)$ two half-band with a common edge
(see Figure \ref{fig:sa}), then
 \[
  \A(D_{(0,d)(b,c)})=\A(B_{L,(a,c)})\cap \A(B_{L,(a,b)}+(d,0))'.
 \]
\end{enumerate}
\end{enumerate}

\begin{figure}[ht]\centering
\begin{tikzpicture}[scale=0.75]
         \draw [thick, ->] (-5,0) --(3,0) node [above left] {$a_1$};
         \draw [thick, ->] (0,-3)--(0,5) node [below right] {$a_0$};
         \draw [ultra thick] (0,0)-- (-3.5,3.5);
         \draw [ultra thick,dotted] (-3.5,3.5)--(-4,4) node [above left] {$B_{\L, (a,1)}$};
         \draw [ultra thick] (1,1)-- (-2.5,4.5);
         \node at(1.2,0.5) {$D_0$};
         \draw [ultra thick] (-2.5,-0.5)-- (0,2);
         \draw [ultra thick,dotted] (-2.5,4.5)--(-3,5);
         \draw [ultra thick] (-1.5,-1.5)-- (-5,2);
         \draw [ultra thick,dotted] (-5,2)--(-5.5,2.5) node [above] {$B_{\L,(a,0)} + (1,0)$};
         \draw [ultra thick,dotted] (2.5,2.5)--(3,3);
         \draw [ultra thick] (2.5,2.5) -- (-2.5,-2.5);
         \draw [ultra thick,dotted] (-2.5,-2.5)--(-3,-3);
         \draw [ultra thick,dotted] (-0.5,-0.5)--(-1,-1);
              \fill [color=black,opacity=0.2]
               (1,1)--(-2.5,4.5)--(-3.5,3.5)--(-1,1)--(-2.5,-0.5)--(-1.5,-1.5)--(0,0);
               \fill [color=black,opacity=0.3]
                (-1,1)--(-3.5,3.5)--(-5,2)--(-2.5,-0.5)--(-1,1);
               \fill [color=black,opacity=0.6]
                (0,0)--(1,1)--(0,2)--(-1,1)--(0,0);
\end{tikzpicture}
\caption{Double cone $D_0$ as the relative causal complement of the shifted half-band
$B_{\L,(a,0)} + (1,0)$ in $B_{\L, (a,1)}$ with $a<0$.}
\label{fig:sa}
\end{figure}
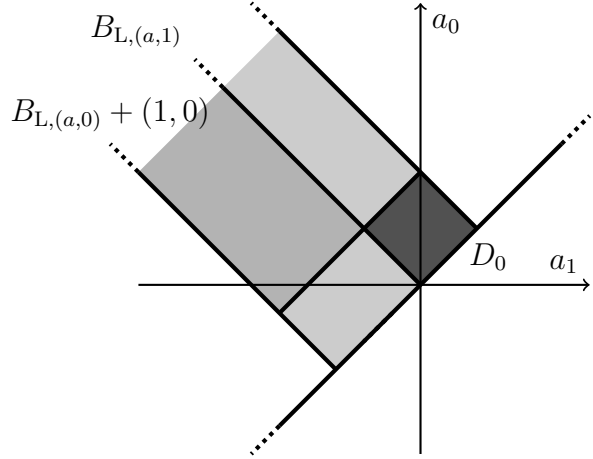

The condition (HK\ref{hsmib}) is concerned with the modular groups of half-bands,
and might look too strong.
On the other hand, (HK\ref{sa}) can be considered as a variation of strong additivity:
indeed, let the net be M\"obius covariant and extend to $\cyl$.
We say the net is strongly additive if $\A(D) = \A(D_1) \vee \A(D_2)$,
where $D, D_1, D_2$ are double cones such that $D_1$ and $D_2$ share one boundary point,
are spacelike to each other and the causal completion of their union is $D$.
It is immediate to see that this condition is equivalent to Haag duality on $M$
(see for the latter e.g.\! \cite[Section 2.1 a)]{KL04-2}, \cite[Proposition 5.2, a)]{CLTW12-2}).
Furthermore, if $\A$ is M\"obius covariant, then Haag duality on $\cyl$ is automatic \cite[Theorem 2.3(i)]{BGL93}.
Therefore, under M\"obius covariance, Haag duality on $M$ is equivalent to (HK\ref{sa}).
For this reason we call (HK\ref{sa}) a (variant of) strong additivity, although $\A$ does not {\it a priori}
extend to $\cyl$.

Differently from (HK\ref{hsmib}), (HK\ref{sa}) does not refer to modular groups which are
without {\it a priori} M\"obius covariance difficult to determine, hence is more transparent for a sufficient condition
for the implication ``dilation + $\a$ $\Longrightarrow$ M\"obius''.

In Proposition \ref{pr:ba} below we show that (HK\ref{sa}) implies (HK\ref{hsmib}).
Only (HK\ref{isotony})--(HK\ref{hsmib}) are needed for the proof of conformal covariance in Theorem \ref{th:mob}.

\section{Proof of M\"obius covariance}\label{proofmc}

\begin{lemma}\label{lem:pxp}
Let $(\A,U,\Omega)$ be a von Neumann algebra net satisfying conditions (HK\ref{isotony})--(HK\ref{bwcone}).
Then $U$ extends to a representation $U_\R$ of the group $\mathrm{\mathbf{P}}\times \uMob$
(and analogously $\mathrm{\uMob\times\mathbf{P}}$)
and local covariance holds for left half-band algebras:
$\Ad U_\R(g)(\A(B_{\L,(a,b)})) = \A(g\cdot B_{\L,(a,b)})$
as long as $g$ is in a neighborhood of the unit element of $\mathrm{\mathbf{P}}\times \uMob$
whose action does not take $B_{\L,(a,b)}$ out of the Minkowski space $M$. 
\end{lemma}
\begin{proof}
We shall denote $\M_1=\A(W_\L + (0,1))$, $\M_2= \A(B_\L)$, $\M_3= \A(V_+) $ and $\Delta_k$, $k=1,2,3$
the associated modular operators w.r.t the vacuum vector $\Omega$.
By (HK\ref{bw}) and (HK\ref{bwcone}), we have that
$ \M_2=\A(B_\L) \subset \A(V_+) = \M_3$ and $\M_2 = \A(B_\L) \subset \A(W_\L + (0,1)) = \M_1$
are $+$- and $-$-HSMI, respectively (see Figure \ref{fig:wedge-cone}).
Therefore, by Lemma \ref{lm:hsmi}, their modular groups $\Delta_2^{it_1}, \Delta_3^{it_3}$
($\Delta_1^{it_1}, \Delta_2^{it_2}$ respectively) satisfy the commutation relations
of $\Lambda_{(0,1)}$ and $\Lambda_{\RR_+}$ (respectively those of $\Lambda_{\RR_-+1}$ and $\Lambda_{(0,1)}$).
Let $\delta(\cdot)\times\iota$ be dilations along the line $a_0+a_1=0$.
They are included in the Poincar\'e-dilation group: indeed,
$\delta(t)\times\iota=\Lambda_{V_+}(-\frac t 2)\Lambda_{W_\L}(-\frac t 2)$, and hence we have
$U(\delta(t)\times\iota)=U(\Lambda_{V_+}(-\frac t 2))U(\Lambda_{W_\L}(-\frac t 2))$.
By Lemma \ref{lem:fix}, all the above modular groups $\Delta_k^{it}$ commute with $U(\delta(t)\times\iota)$, since 
the latter preserves each algebra and the vacuum.
Note also that
\begin{align}\label{eq:commutation2}
 &U(\delta(t)\times\iota)U(\delta(s)\times\iota) \nonumber\\
 &= 
 U\left(\delta\left(\ln({e^{t+s}+1-e^{t}})\right)\times\iota\right)
 U\left(\delta\left(\ln\left(\frac{e^{t+s}}{{e^{t+s}+1-e^{t}}}\right)\right)\times\iota\right)
\end{align}
as $U(\delta(\cdot)\times\iota)$ is a one-parameter group. In particular,
they satisfy the commutation relations of $\Lambda_{(0,1)}$ and $\Lambda_{\RR_+}$.
Similarly, it also satisfies the commutation relations of $\Lambda_{(0,1)}$ and $\Lambda_{\RR_-+1}$.
Therefore, by straightforward computations, we have the following:
\begin{itemize}
 \item $\Delta_2^{it_2}U(\delta(2\pi t_2)\times\iota)$ and $\Delta_3^{it_3}U(\delta(2\pi t_3)\times\iota) = U(\iota\times\delta(-2\pi t_3))$
 satisfy the commutation relations of $\Lambda_{(0,1)}$ and $\Lambda_{\RR_+}$.
 \item $\Delta_2^{it_2}U(\delta(2\pi t_2)\times\iota)$ and $\Delta_1^{it_1}U(\delta(2\pi t_1)\times\iota) = U(\iota\times\Lambda_{\RR_-+1}(2\pi t_1))$
 satisfy the commutation relations of  $\Lambda_{(0,1)}$ and $\Lambda_{\RR_-+1}$.
\end{itemize}

\begin{figure}[ht]\centering
\begin{tikzpicture}[scale=0.75]
         \draw [thick, ->] (-5,0) --(3,0) node [above left] {$a_1$};
         \draw [thick, ->] (0,-3)--(0,5) node [below right] {$a_0$} node [below left] {$V_+$};
         \draw [ultra thick] (0,0)-- (-3.5,3.5);
         \node at(-3.6,4.5) {$B_{\L}$};
         \node at(-4,1.5) {$W_\L$};
         \draw [ultra thick,dotted] (-3.5,3.5)--(-4,4);
         \draw [ultra thick] (1,1)-- (-2.5,4.5);
         \draw [ultra thick,dotted] (-2.5,4.5)--(-3,5);
         \draw [ultra thick,dotted] (3.5,3.5)--(4,4);
         \draw [ultra thick] (3.5,3.5) -- (-3.5,-3.5);
         \draw [ultra thick,dotted] (-3.5,-3.5)--(-4,-4);
         \draw [ultra thick,dotted] (-0.5,-0.5)--(-1,-1);
               \fill [color=black,opacity=0.3]
                 (0,0)--(3.5,3.5)--(-3.5,3.5)--(0,0);
               \fill [color=black,opacity=0.6]
                (1,1)--(-2.5,4.5)--(-3.5,3.5)--(0,0);
               \fill [color=black,opacity=0.4]
                (0,0)--(-3.5,3.5)--(-3.5,-3.5)--(1,1);
\end{tikzpicture}
\caption{Regions $V_+, B_\L, W_\L + (1,1)$ and their shadows $\RR_+, (0,1), \RR_- +1$ on the line $a_0=a_1$.}
\label{fig:wedge-cone}
\end{figure}
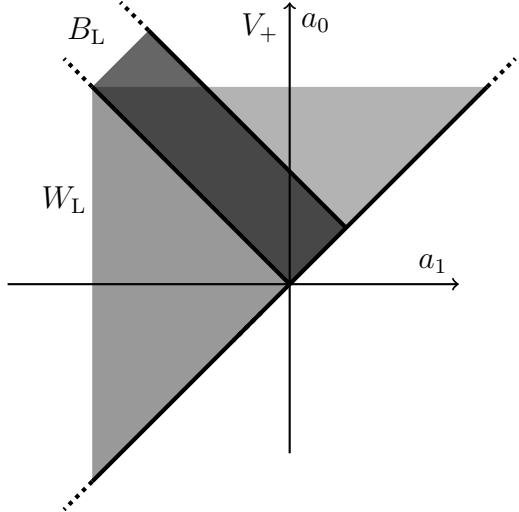

On the other hand, it is immediate that
also $U(\iota\times\delta(-2\pi t_3)) = U(\iota\times\Lambda_{\RR_+}(2\pi t_3))$ and $U(\iota\times\Lambda_{\RR_-+1}(2\pi t_1))$
satisfy the commutation relations of $\Lambda_{\RR_+}$ and $\Lambda_{\RR_-+1}$.
Therefore, by Lemma \ref{lm:group},
we obtain a strongly continuous representation $U_\R$ of $\uMob$ which coincides with $U(\iota\times\cdot\,)$
when restricted to $\iota\times\mathrm{\mathbf{P}}$ by construction.

We claim that this representation commutes with $U(g\times\iota), g\in \mathrm{\mathbf{P}}$.
We only have to show that $\Delta_2^{it}U(\delta(2\pi t)\times\iota)$
commutes with $U(g\times\iota)$.
By Lemma \ref{lem:fix}, $\Delta_2^{it}$ and hence $\Delta_2^{it}U(\delta(2\pi t)\times\iota)$
commute with
$U(\delta(s)\times\iota)$ since $\Ad U(\delta(2\pi t)\times\iota)(\A(B_\L)) = \A(B_\L)$
and $U(\delta(s)\times\iota)\Omega = \Omega$.
Furthermore, as $\Ad U(\tau(a)\times\iota)(\A(B_\L)) \subset \A(B_\L)$ for $a \ge 0$
and $U(\tau(a)\times\iota)$ preserves $\Omega$,
by Theorem \ref{thm:bor} 
we have that $\Ad \Delta_2^{it}(U(\tau(a)\times\iota)) = U(\tau(e^{-2\pi t}a)\times\iota)$.
Moreover it also holds that $\Ad U(\delta(2\pi t_2)\times\iota)(U(\tau(a)\times\iota))  = U(\tau(e^{2\pi t}a)\times\iota)$
as it is a representation of $\textbf{P}$. It follows that $\Delta_2^{it}U(\delta(2\pi t)\times\iota)$
commutes with $U(\tau(a)\times\iota)$.

Altogether, $U(g_1\times \iota)$ and $U_\R(\iota\times g_2)$ commute for
$g_1 \in \mathrm{\mathbf{P}}, g_2 \in \uMob$.
We define a representation $U_\R$ of $\mathrm{\mathbf{P}} \times \uMob$
by $U_R(g_1\times g_2) := U(g_1)U_\R(g_2)$.

Now we prove local covariance of $\{\A(B_{\L,(a,b)})\}$ with respect to $U_\R$.
It is enough to check it for $\iota\times g\in \iota\times\uMob$, because covariance for elements in $\mathrm{\mathbf{P}}\times\iota$
is the assumption (HK\ref{poincare}) and (HK\ref{dilation}).
We view that $\uMob$ acts on the universal covering $\widetilde{S^1}$ of $S^1$
(which is homeomorphic to $\RR$, see Figure \ref{fig:circlecover}).
Let us denote by $I_{(a,b)}$ the interval in $\widetilde{S^1}$ corresponding to $(a,b)$ in $\RR$.
We also identify $\RR$ with an interval $\widetilde{S^1}$ and denote it by $I_\RR$.
{\it In this paragraph, we consider $\A$ as a net defined on
open regions in $I_\RR \times I_\RR$}.
Let $I \Subset I_\RR$ be a bounded interval. Any $g \in \uMob$ such that $gI \Subset I_\RR$ can be written
as a product of three elements $g=g_1 g_2 g_3$, such that $g_1,g_3\in \textbf{P}$, $g_3 I=I_{(0,1)}$, $g_1\cdot I_{(0,1)} = gI$
and $g_2=\Lambda_{(0,1)}(t_2)$ for some $t_2\in\RR$. 
For such $g \in \uMob$, note that
$U_\R(\iota\times g_1)$ and $U_\R(\iota\times g_3)$ acts geometrically by (HK\ref{poincare}) and (HK\ref{dilation}),
and $U_\R(\iota\times g_2) = \Delta_2^{it_2}U(\delta(2\pi t_2)\times\iota)$
preserves $\A(B_\L) = \A(I_{\RR_+}\times g_3 I) = \A(I_{\RR_+}\times I_{(0,1)})$,
therefore, we have
\begin{align*}
 \Ad U_\R(\iota\times g)(\A(I_{\RR_+}\times I)) &=  \Ad U_\R(\iota\times g_1g_2g_3)(\A(I_{\RR_+}\times I)) \\
 &=  \Ad U_\R(\iota\times g_1g_2)(\A(I_{\RR_+}\times g_3 I)) \\
 &=  \Ad U_\R(\iota\times g_1)(\A(I_{\RR_+}\times g_3 I)) \\
 &= \A(I_{\RR_+}\times gI),
 \end{align*}
which is the desired local covariance.
\begin{figure}[ht]\centering
\begin{tikzpicture}[path fading=north,scale=0.75]
         \draw [thick] (-6,0) --(6,0);
         \draw [thick,dotted] (-6,0) --(-7,0);
         \draw [thick,dotted] (6,0) --(7,0);
         \node at(0,1) {$\RR$};
         \draw [ultra thick] (-3,0) node{$($}--(-0.5,0)node{{\footnotesize$($}}--(0.5,0)node[below]{$(a,b)$}--(1.5,0)node{{\small$)$}}--(3,0)node{$)$};
\end{tikzpicture}
\caption{The universal covering of $S^1$, and $\RR$ as an interval on it.}
\label{fig:circlecover}
\end{figure}
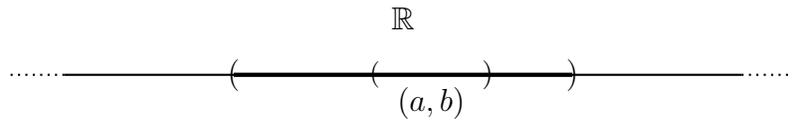
At this point, we can prove {\bf continuity from below} (c.f.\! \cite[(24)]{FJ96}) for left half-bands, namely,
$\A(\RR_+\times I) = \bigvee_{I_\a \Subset I}\A(\RR_+\times I_\a)$.
Indeed, for each $I_\a$ we can find $g_\a \in \uMob$ such that $g_\a I =  I_\a$
and $g_\a \to \i$ as $I_\a$ tends to $I$.
Now by continuity of $U_\R$ in the strong operator topology,
$\bigvee_{I_\a \Subset I}\A(I_{\RR_+}\times I_\a) = \bigvee_\a \Ad U_\R(\iota\times g_\a) \A(I_{\RR_+}\times I)
\supset \A(I_{\RR_+}\times I)$,
and the converse inclusion is trivial.

Now we bring back the original notations and $I, I_\a$ are intervals in $\RR$.
As for wedges, we have by definition that $\A(W_\L) = \bigvee_{I \Subset \RR_-} \A(\RR_+\times I)$.
Therefore, for $g$ which takes any interval of $\RR_-$ into a compact interval in $\RR$, we have
$\Ad U_\R(\iota\times g) (\A(W_{L})) = \bigvee_{I \Subset \RR_-} \A(\RR_+\times gI) = \A(\RR_+\times g \RR_-)$,
where the last equality follows from continuity from below.
In other words, covariance holds also for $W_\L$.
By taking $g = g_1g_2$ where $g_2 \in \mathrm{\mathbf{P}}$,
covariance holds for any $W_\L + (0, a_\R)$, $a_\R\in\RR$.
Finally, as $U_\R(\tau(a_\L)\times\iota)$ commutes with $U_\R(\iota\times g)$, the above covariance holds for any left wedge.

\end{proof}

\begin{proposition}\label{pr:ba} Let $(\A,U,\Omega)$ be a von Neumann algebra net satisfying conditions
(HK\ref{isotony})--(HK\ref{bwcone})(HK\ref{sa}). Then it also satisfies (HK\ref{hsmib}).
\end{proposition}
\begin{proof}
By (HK\ref{sa}), it holds that $\A(D_0) = \A(\RR_+\times (a_\L, 1))\cap \A((\RR_+ +1)\times (a_\L, 0))'$
for any ${a_\L} < 0$. For $t > 0$, $\Lambda_{(0,1)}(t)$ takes any interval in $\RR_-$ in $\RR_-$ itself,
and say, $\Lambda_{(0,1)}(t)\cdot a_\L = b_\L < 0$, and $\Lambda_{(0,1)}(t) \cdot 1 = 1, \Lambda_{(0,1)}(t) \cdot 0 = 0$.
By covariance of Lemma \ref{lem:pxp}, we have
\[
 \Ad U_\R(\iota\times\Lambda_{(0,1)}(t))(\A(D_0)) = \A(\RR_+\times(b_\L, 1))\cap \A((\RR_+ + 1)\times (b_\L, 0))'
= \A(D_0),
\]
where the last equality holds for any negative number $b_\L$ by (HK\ref{sa}).
Now we can invert $t$ and the equality holds also for $t \le 0$.
Recalling the definition $U_\R(\iota\times\Lambda_{(0,1)}(t)) = \Delta_{B_\L}^{it}U(\delta(2\pi t)\times\iota)$,
and that $\Ad U_\R(\iota\times\Lambda_{(0,1)}(t))$ preserves $\A(D_0)$,
we conclude that $\Ad U(\delta(-2\pi t)\times\iota)(\A(D_0)) = \Ad \Delta_{B_\L}^{it}(\A(D_0))$.
\end{proof}

\begin{theorem}\label{th:mob}
Let $(\A,U,\Omega)$ be a von Neumann algebra net satisfying conditions (HK\ref{isotony})--(HK\ref{hsmib}).
Then $U$ extends to the two-dimensional M\"obius group $G$ and with this extension
$(\A, U, \Omega)$ is a M\"obius covariant net.
\end{theorem}
\begin{proof}
Let us similarly define $U_\L$
(the only assumption which is not symmetric between left and right is (HK\ref{hsmib}),
but we have not used it for Lemma \ref{lem:pxp}).
Namely, $U_\L(\Lambda_{(0,1)}(2\pi s)\times\iota) = \Delta_{B_\R}^{is}U(\iota\times\delta(2\pi s))$.
We are going to prove that $U_\R(\iota\times\Lambda_{(0,1)}(2\pi t))$ and $U_\L(\Lambda_{(0,1)}(2\pi s)\times\iota)$ commute,
and hence we will have a representation of $\uMob\times\uMob$.

First, we show that $U(\iota\times\delta(2\pi t))$ commutes with $\Delta_{D_0}^{is}U_\R(\iota\times\Lambda_{(0,1)}(-2\pi s))$.
Note that
\begin{itemize}
\item From (HK\ref{hsmib}) and Lemma \ref{lem:fix},
$U_\R(\iota\times\Lambda_{(0,1)}(2\pi s))$ commutes with the modular group $\Delta_{D_0}^{it}$ of $\A(D_0)$.
\item By (HK\ref{bwcone}) and (HK\ref{hsmib}) 
respectively,
namely $\Delta_{V_+}^{it} = U(\delta(-2\pi t)\times\delta(-2\pi t))$
and $\Delta_{B_\L}^{is} = U_\R(\delta(-2\pi s)\times\Lambda_{(0,1)}(2\pi s))$
(see Lemma \ref{lem:pxp}),
the following three are all $+$-HSMI (see Figure \ref{fig:triple-hsmi}):
\[
\A(D_0) \subset \A(V_+), \quad
\A(D_0) \subset A(B_\L), \quad
\A(B_\L)\subset \A(V_+).
\]
\item The same relation as \eqref{eq:commutation2} holds for any one-parameter group.
\end{itemize}
Therefore, by putting $s_1 =  \ln({e^{t+s}+1-e^t}), t_1= \ln\left(\frac{e^{t+s}}{{e^{t+s}+1-e^t}}\right)$
for small $t$ we have
\begin{align*}
 &U(\iota\times\delta(-2\pi t))\Delta_{D_0}^{is}U_\R(\iota\times\Lambda_{(0,1)}(-2\pi s))U(\iota\times\delta(2\pi t)) \\
 &= U(\delta(2\pi t)\times\iota)\Delta_{V_+}^{it}\Delta_{D_0}^{is}\Delta_{B_L}^{-is}U(\delta(-2\pi s)\times\iota)\Delta_{V_+}^{-it}U(\delta(-2\pi t)\times\iota) &\text{(def.\! of } U_\R, \text{ (HK\ref{bwcone}))}\\
 &= U(\delta(2\pi t)\times\iota)\, \Delta_{V_+}^{it}\Delta_{D_0}^{is}\Delta_{B_\L}^{-is}\Delta_{V_+}^{-it}\, U(\delta(-2\pi s)\times\iota)U(\delta(-2\pi t)\times\iota) &\text{(reordering)}\\
 &= U(\delta(2\pi t)\times\iota)\, \Delta_{D_0}^{is_1}\Delta_{V_+}^{it_1}\Delta_{V_+}^{-it_1}\Delta_{B_\L}^{-is_1}\, U(\delta(-2\pi s)\times\iota)U(\delta(-2\pi t)\times\iota) &\text{(by Eq.\! \eqref{eq:commutation})}\\
 &= U(\delta(2\pi t)\times\iota)\, \Delta_{D_0}^{is_1}\Delta_{B_\L}^{it_1}\Delta_{B_\L}^{-it_1}\Delta_{B_\L}^{-is_1}\, U(\delta(-2\pi s)\times\iota)U(\delta(-2\pi t)\times\iota) &\text{(cancelling factors)} \\
 &= U(\delta(2\pi t)\times\iota)\, \Delta_{B_\L}^{it}\Delta_{D_0}^{is}\Delta_{B_\L}^{-is}\Delta_{B_\L}^{-it}\, U(\delta(-2\pi s)\times\iota)U(\delta(-2\pi t)\times\iota) &\text{(by Eqs.\! \eqref{eq:commutation}\eqref{eq:commutation2})}\\
 &= U_\R(\iota\times\Lambda_{(0,1)}(2\pi t))\, \Delta_{D_0}^{is}\, U_\R(\iota\times\Lambda_{(0,1)}(2\pi (-s-t))) &(\text{def.\! of } U_\R)\\
 &= \Delta_{D_0}^{is}\, U_\R(\iota\times\Lambda_{(0,1)}(-2\pi s)). &([U_\R(\Lambda_2(t)), \Delta_{D_0}^{is}] = 0)
\end{align*}
But if this is valid for small $t$, it is valid also for any $t$ by iteration.

We claim that $\Delta_{D_0}^{it} = U_\R(\iota\times\Lambda_{(0,1)}(2\pi t)) U_\L(\Lambda_{(0,1)}(2\pi t)\times\iota)$.
One one hand, we know from the previous paragraph that
\[
 \Ad U(\iota\times\delta(2\pi s))(\Delta_{D_0}^{it}U_\R(\iota\times\Lambda_{(0,1)}(-2\pi t)))
 = \Delta_{D_0}^{it}U_+(\iota\times\Lambda_{(0,1)}(-2\pi t)).
\]
On the other hand,
by Theorem \ref{thm:bor} and Proposition \ref{lem:pxp}, we have
\[
 \Ad U(\iota\times\delta(2\pi s))(\Delta_{D_0}^{it}U_\R(\iota\times\Lambda_{(0,1)}(-2\pi t)))
 = \Delta_{\iota\times\delta(2\pi s)\cdot D_0}^{it}U_\R(\iota\times\Lambda_{(0,s)}(-2\pi t)).
\]
Combining these two equalities and the limit $s\to\infty$,
we obtain $\Delta_{D_0}^{it}U_\R(\iota\times\Lambda_{(0,1)}(-2\pi t)) = \Delta_{B_R}^{it}U(\iota\times\Lambda_{\RR_+}(-2\pi t))$.
Recall that $\Delta_{B_R}^{it} = U(\iota\times\delta(-2\pi t))U_\L(\Lambda_{(0,1)}(t)\times\iota)$ by definition,
and we have $U(\iota\times\delta(-2\pi t)) = U(\iota\times\Lambda_{\RR_+}(2\pi t))$,
hence, $\Delta_{D_0}^{it} = U_\R(\iota\times\Lambda_{(0,1)}(2\pi t)) U_\L(\Lambda_{(0,1)}(2\pi t)\times\iota)$. 
In particular, $U_\R(\iota\times\Lambda_{(0,1)}(t))$ and $U_\L(\Lambda_{(0,1)}(s)\times\iota)$ commute
and we obtain a representation $U$
of $\uMob \times \uMob$ by $U(g_\L\times g_\R) = U_L(g_\L\times\iota)U_\R(\iota\times g_\R)$.
\begin{figure}[ht]\centering
\begin{tikzpicture}[path fading=north,scale=0.75]
         \draw [thick, ->] (-4,0) --(4,0) node [above left] {$a_1$};
         \draw [thick, ->] (0,-1)--(0,5) node [below right] {$a_0$};
         \draw [ultra thick] (0,0)-- (-3.5,3.5);
         \node at(-2,3) {$B_\L$};
         \draw [ultra thick,dotted] (-3.5,3.5)--(-4,4);
         \draw [ultra thick] (1,1) -- (-2.5,4.5);
         \node at(1.2,0.5) {$D_0$};
         \draw [ultra thick] (-1,1)-- (0,2);
         \node at(1,3.5) {$V_+$};
         \draw [ultra thick,dotted] (-2.5,4.5)--(-3,5);
         \draw [ultra thick,dotted] (3.5,3.5)--(4,4);
         \draw [ultra thick] (3.5,3.5)-- (-0.5,-0.5);
         \draw [ultra thick,dotted] (-0.5,-0.5)--(-1,-1);
              \fill [color=black,opacity=0.2]
               (0,0)--(1,1)--(0,2)--(-1,1)--(0,0);
              \fill [color=black,opacity=0.2]
               (0,0)--(1,1)--(-3,5)--(-4,4)--(0,0);
              \fill [color=black,opacity=0.3]
               (0,0)--(1,1)--(4,4)--(-4,4)--(0,0);
\end{tikzpicture}
\caption{Triple HSMI, $D_0\subset B_\L, B_\L\subset V_+, D_0\subset V_+$.}
\label{fig:triple-hsmi}
\end{figure}
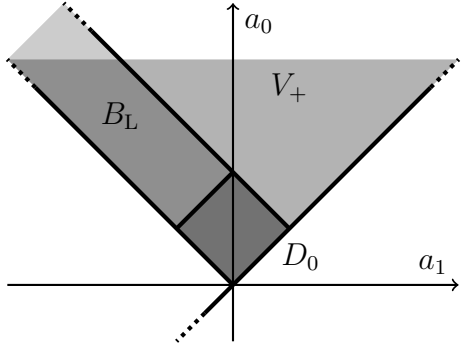

Now we prove local covariance of double cone algebras under $\uMob\times\uMob$ action.
By dilation and translation covariance, it is enough to consider $D_0$.
Let $I_\L$ and $I_\R$ be intervals on the line $\RR$ such that $D_0=I_\L\times I_\R$.
Recall that in the last step we proved $\Delta_{D_0}^{it} = U_+(\Lambda_{(0,1)}(-2\pi t))U_-(\Lambda_{(0,1)}(-2\pi t))$.
Now let us take an element of the form $g\times \iota \in \uMob\times \uMob$,
such that 
$g_\L=g_{\L,1}g_{\L,2}, g_\R=g_{\R,1}g_{\R,2}$, where
$g_{\L,1}, g_{\R,1} \in \mathrm{\mathbf{P}}, g_{\L,1}=\Lambda_{(0,1)}(t), g_{\R,1}=\Lambda_{(0,1)}(s)$ for some $t,s\in\RR$.
For such element, $\Ad U(g_\L\times g_\R)(\A(D_0)) = \Ad U(g_{\L,1}\times g_{\R,1})U(g_{\L,2}\times g_{\R,2})(\A(D_0))
= \Ad U(g_{\L,1}\times g_{\R,1})(\A(D_0)) = \A((g_{\L,1}\times g_{\R,1})\cdot D_0)$,
because $U(g_{\L,2}\times g_{\R,2}) = U(\Lambda_{(0,1)}(t)\times \Lambda_{(0,1)}(s))$ preserves $\A(D_0)$
and covariance for $g_{\L,1}\times g_{\R,1} \in \mathrm{\mathbf{P}}\times \mathrm{\mathbf{P}}$ holds by assumptions.
As any element $g \in \uMob\times\uMob$ which does not take $D_0$ outside the Minkowski space $M$
can be written as above, this establishes local covariance with respect to $\uMob\times\uMob$.

From here, by the conformal spin-statistic theorem (Theorem \ref{th:ss}), $U$ factors through $G$ and
we conclude that $(\A,U,\Omega)$ is a M\"obius covariant net.
\end{proof}

\begin{corollary}\label{cr:mob}
Let $(\A, U, \Omega)$ be a conformal net satisfying (HK\ref{isotony})--(HK\ref{bwcone}) and (HK\ref{sa}).
Then $U$ extends to the two-dimensional M\"obius group $G$ and with this extension
$(\A, U, \Omega)$ is a M\"obius covariant net.
\end{corollary}
\begin{proof}
 Immediate from Proposition \ref{pr:ba} and Theorem \ref{th:mob}.
\end{proof}

\section{Counterexamples}\label{counterexample}

In this section we discuss several Haag-Kastler nets which are covariant with respect to the Poincar\'e-dilation group
but cannot be extended to a M\"obius covariant net.
Constructions and results of these sections can be easily adapted to $(3+1)$-dimensions.
One should substitute the representation of $\uMob\times\uMob$ in Section \ref{nochiral} by ``massive'' representations of the conformal group
in the sense of \cite{Mack77}, and note that any double cone $O$ can be
obtained as the intersection of countably many wedges.

\subsection{Breaking the Bisognano-Wichmann property for future light cone}\label{internal}

The Bisognano-Wichmann property for boosts holds in any M\"obius covariant net on $(1+1)$-dimensional the Minkowski spacetime \cite[Theorem 2.3]{BGL93}.
Furthermore, since wedges are mapped to lightcones by M\"obius transformations,
the Bisognano-Wichmann property for $V_+$ holds in a M\"obius covariant net
(corresponding results hold for any conformally covariant net on higher-dimensional Minkowski space).

Let $(\A, U, \Omega)$ be a M\"obius covariant net which admits a one-parameter inner symmetry,
namely, there is a one-parameter unitary group $\{V(s)\}$ such that $\Ad V(s)\A(O) = \A(O)$
and $[U(g), V(s)] = 0$ for $g \in G$ and $s \in \RR$ and $V(s)\Omega = \Omega$.
In this situation, we can construct a new representation of the Poincar\'e-dilation group by setting
$U_V(g) = U(g)$ for $g \in \poincare$ and $U_V(\delta(t)\times\delta(t)) = U(\delta(t)\times\delta(t))V(t)$.
It is easy to check that $U_V$ is still a representation of the Poincar\'e-dilation group,
$\A$ is covariant under $U_V$ and $\Omega$ is invariant under $U_V$.
While the original net $(\A, U, \Omega)$ satisfies the Bisognano-Wichmann property,
the new net $(\A, U_V, \Omega)$ violates it since the algebra $\A(V_+)$ and the vacuum $\Omega$
stay the same while $U_V(\delta(t)\times\delta(t))$ has been modified.
Therefore, this $U_V$ cannot be extended to $G$ in such a way that $\A$ is still covariant,
because that would contradicts the Bisognano-Wichmann property for $V_+$ \cite[Theorem 2.3]{BGL93}.

It is easy to find such M\"obius covariant nets.
For example, if $(\A_0, U_0, \Omega_0)$ is a M\"obius covariant net on $S^1$ with a one-parameter inner symmetry $V_0(s)$,
one can just take the tensor product $\A(I_\L\times I_\R) := \A_0(I_\L)\otimes \A_0(I_\R),
U(g_\L\times g_\R) := U_0(g_\L)\otimes U_0(g_\R), \Omega := \Omega_0\otimes \Omega_0$,
and $V(s) := V_0(s)\otimes V_0(s)$.
As concrete examples, one can take a loop group net $\A_{G,k}$ with a compact group $G$ at level $k$,
and as $V_0$ one can just take any one-parameter group in $G$ which acts as inner symmetry \cite[Section III]{GF93}.
One can also consider the tensor product of two copies of the $\mathrm{U}(1)$-current net as a chiral component
(see \cite[Section 5]{Tanimoto14-1} for the construction of inner symmetry on the complex massive free field, and \cite[Section 5.2]{BT15} for restricting it to a lightray to obtain a net on $S^1$).

As for the converse, we do not know whether it is always possible to satisfy (HK\ref{bwcone})
by modifying $U$ of a given Poincar\'e-dilation covariant net with (HK\ref{rsv}),
see the discussions in Section \ref{oa}.
The fact that it is possible to modify the representation of the Poincar\'e-dilation group is 
known in the physics literature \cite[Section 2.2, below (2.6)]{Nakayama15}.
Finally, we remark that it is also easy in $(1+1)$-dimension to violate the Bisognano-Wichmann property for wedges \cite[Section 5]{Tanimoto14-1}.

\subsection{BGL construction and generalized free fields}\label{generalized}
The simplest two-dimensional quantum field theory is the massive free field.
It cannot be dilation covariant because it has an isolated mass shell.
Yet, if one glue together continuously many massive free fields, the mass spectrum becomes continuous
and dilation may act on it. This idea can be indeed realized as generalized free field.
Here we take the construction based on the one-particle representation of the Poincar\'e group:
to a positive energy (anti-)unitary representation of the Poincar\'e group with the CPT transformation,
Brunetti, Guido and Longo associated a net of real subspaces on wedges,
and its second quantized net \cite{BGL02}.

When the given representation of the Poincar\'e group extends to M\"obius group $G = (\uMob\times\uMob) / \ZZ_2$,
there are two choices for double cones: either one defines the real subspaces for double cones by covariance \`a la BGL,
or by duality for wedges.
We take a representation of $G$ such that
the former construction yields a M\"obius covariant net without chiral components,
while the latter contains a generalized free field and fails to have both M\"obius covariance and the split property.

Generalized free fields can have have dilation covariance but fail to be M\"obius covariant, see e.g.\! \cite[Section 3.1]{DR09}.
In this case, a generalized free field  has scaling dimension which is
not consistent with the unitarity condition of the M\"obius group (or the conformal group in the case of $(3+1)$-dimensions),
and consequently, the field cannot be M\"obius covariant. 
This is different from our counterexamples. We consider first a M\"obius covariant net
(which corresponds to a generalized free field with unitary scaling dimension),
and show that its dual net is dilation covariant but not M\"obius covariant.

\subsubsection{Nets of standard subspaces and first quantization nets}\label{standard}
A general reference for this section is \cite{Longo08}.
Firstly, we recall the notion of real subspaces. A linear, real, closed subspace $H$ of a complex Hilbert space $\H$ is called {\bf cyclic} if 
$H+iH$ is dense in $\H$, {\bf separating} if $H\cap iH=\{0\}$ and 
{\bf standard} if it is cyclic and separating.

If $H$ is a  real linear subspace of $\H$, the \emph{symplectic complement} of $H$ is defined by
\[
H' \equiv \{\xi\in\H\ ;\ \Im\<\xi,\eta\>=0, \text{for } \eta\in H\} = (iH)^{\bot_\RR}\ ,
\]
where $\bot_\RR$ denotes the orthogonal in $\H$  with respect to the real part of the scalar product on $\H$.
$H'$ is a closed, real linear subspace of $\H$. If $H$ is standard, then $H = H''$. 
$H$ is cyclic (respectively separating) if and only if $H'$ is separating (respectively cyclic), thus $H$ is standard if and only if $H'$ is standard.
The {\bf Tomita operator} $S_H$ associated to a standard subspace $H\subset\H$ is the densely defined closed anti-linear involution $H+iH\ni \xi + i\eta \mapsto \xi - i\eta\in H+iH$. The polar decomposition $S_H = J_H\Delta_H^{1/2}$ defines the positive self-adjoint {\bf modular operator} $\Delta_H$ and the anti-unitary {\bf modular conjugation} $J_H$. In particular, $\Delta_H$ is invertible and
\begin{equation}\label{eq:1}
J_H\Delta_H J_H=\Delta_H^{-1}.
\end{equation} We further have that
$S_{H'} = S^*_H$,
 $J_H H = H'$ and $\Delta_H^{it}H = H$  for every $t\in\RR.$
The one-parameter, strongly continuous group $t\mapsto \Delta_H^{it}$ is the {\bf modular group} of $H$. There is a 1--1 correspondence between Tomita operators and  standard subspaces, namely between
standard subspaces $H\subset\H$, operators $S$ which are closed, densely defined anti-linear involutions on $\H$ and pairs $(J,\Delta)$ of an anti-unitary involution $J$ and  a positive self-adjoint operator $\Delta$ on $\H$ satisfying \eqref{eq:1}.
We recall the following analogue of Takesaki's theorem \cite[Proposition 2.1.10]{Longo08}.
\begin{lemma}\label{inc}.
Let $K\subset H\subset \H$ be an inclusion of  standard subspaces in $\H$.
If $\Delta_H^{it}K=K$ for every $t\in\RR$, then $K=H$.
\end{lemma}

Let us denote the set of (left and right) wedges by $\W$ in the two-dimensional Minkowski space $M$.
The one parameter group of Lorentz boosts $\Lambda_{W_\L}(t)=\delta(-t)\times\delta(t)$ is associated to $W_\L$.
We associate to a general left wedge $W=g W_\L$ the one-parameter group $\Lambda_W(t)=g\Lambda_{W_\L}(t)g^{-1}$,
and to a right wedge $W=g W_\R$ the reversed one-parameter group $\Lambda_W(t)=g\Lambda_{W_\L}(-t)g^{-1}$.

Let $\poincarej$ be the proper Poincar\'e group, namely, it is generated by the connected component $\poincare$
of the full Poincar\'e group and $j: (a_\L, a_\R) \mapsto (-a_\L, -a_\R)$. Let $\alpha$ be the action of $j$ on $\Poi$, then $\poincarej=\Poi\rtimes_\alpha\ZZ_2$. 
We introduce also $j_W := g j g^{-1}$, where $g\in \poincare$ is such that $W = gW_\L$ or $gW_\R$,
depending on whether $W$ is a left or right wedge (this does not depend on the choice of $g$, hence is well-defined).
Let $U$ be a (anti-)unitary, positive energy representation of $\poincarej$, namely unitary on $\Poi$ and anti-unitary on $j_W\Poi$. In particular $j_W$ is represented by an anti-unitary operator $U(j_W)$. 
Define $\Delta_W$ by the equation $U(\Lambda_W(t)) = \Delta_W^{i2\pi t}$.
Brunetti, Guido and Longo associated a standard real subspace $H(W) = \ker (\1 - U(j_W)\Delta_W^\frac12)$.
They form a ($\poincarej$-covariant) {\bf net of standard subspaces for wedges}
in the following sense \cite[Theorem 4.7]{BGL02}:
 \begin{enumerate}[{(SS}1{)}]
\item\label{ssisotony} \textbf{Isotony:}  If $W_1,W_2$ and $W_1\subset W_2$ then $H(W_1)\subset H(W_2)$;
\item \textbf{Poincar\'e Covariance:}   $U(g)H(W)=H(gW)$ for $g\in\poincarej,\, W\in\W$;
\item \textbf{Positivity of energy:} the joint spectrum of translations in $U$ is contained in $\overline{V_+}$;
\item \textbf{Reeh-Schlieder property:} $H(W)$ is standard in $\H$ for any $W\in\W$
\item \textbf{Locality:}  for any $W_1\subset W_2'$ then $H(W_1)\subset H(W_2)'$,
where $W'$ denotes the causal complement of $W$ in the Minkowski space $M$
\item\label{SSbw} \textbf{Bisognano-Wichmann property:} $U(\Lambda_W(t)) = \Delta_W^{i2\pi t}$ for every $W\in\W$ and $t\in\RR$.
\end{enumerate}

Given a net of standard subspaces $\{H(W)\}$ for wedges, the {\bf dual net} $H^\dual$ of real subspaces is defined for a double cone $O$ by
\begin{align}\label{eq:HOdual}
 H^\dual(O)\dot=H(O')',
\end{align}
where the causal complement $O'$ of $O$ in the Minkowski space $M$ consists of two wedges $W_1,W_2$,
and $H(O')$ is the real closed subspace spanned by $H(W_1)$ and $H(W_2)$.
We do not know for which class of $U$ $H(O)$ is standard, but we can prove it for a concrete class in Section \ref{nochiral}.
In that case, the net of standard subspaces $\{H^\dual(O)\}$ satisfies
isotony, $\poincarej$-covariance, positivity of energy, the Reeh-Schlieder property and locality
in the natural sense.

If $U$ extends to the group $G\rtimes_\alpha \ZZ_2$ (where $G$ is the two-dimensional M\"obius group),
there is another choice to assign a standard subspace to double cones.
Let us define $H(O) = \ker(\1-U(j_O)\Delta_O^\frac12)$, where $U(j_O) = U(g)U(J_W)U(g)^*$ with $g$ such that $O = gW$
and $\Delta_O = U(g)\Delta_W U(g)^*$.
With this definition, the net of standard subspaces $\{H(O)\}$ satisfies
isotony, $G\rtimes_\alpha \ZZ_2$-covariance, positivity of energy, the Reeh-Schlieder property and locality
in the natural sense where $O$ is a double cone on $\cyl$.
We refer to this the {\bf BGL construction} of standard subspaces associated with $U$.

\paragraph{Bosonic second quantization.} Let $\H$ be a (complex) Hilbert space and $\F_+(\H)$ be the associated bosonic Fock space.
Given a real subspace $H\subset\H$, we shall denote with ${\mathcal R}_+(H)$ the second quantization von Neumann algebra 
 \[
  {\mathcal R}_+(H)=\{W_+(f):f\in H\}''\subset\B(\F_+(\H))
 \]
where $W_+(f)$ are the Weyl operator\footnote{The symbol $W_+$ should not be confused with wedges $W$.
Weyl operators appear only in this section.}
on the Fock space, satisfying the CCR 
 \[
  W_+(f)W_+(g)=e^{\Im(f,g)}W_+(f+g),\qquad f,g\in H)
 \]
and $\omega (W_+(f)) =(\Omega,W_+(f)\Omega)=e^{-\frac12 \|f\|^2}$.
The second quantization construction respects the lattice structure and the modular theory. Let $\Gamma_+(A)$ be the multiplicative Bose second
quantization of a one-particle operator $A$ on $\H$.
\begin{proposition}\label{prop:secquant} \cite{Araki63,LRT78,LMR16}
Let $H$ and $\{H_\k\}$ be closed, real linear subspaces of $\H$. We have
\begin{itemize}\itemsep0mm
\item[$(a)$] ${\mathcal R}_+(H)' = {\mathcal R}_+(H')$; 
\item[$(b)$] ${\mathcal R}_+(\sum_\k H_\k) = \bigvee_\k {\mathcal R}_+(H_\k)$;
\item[$(c)$] ${\mathcal R}_+(\bigcap_\k H_\k) = \bigcap_\k {\mathcal R}_+(H_\k)$.
\item[$(d)$] If $H$ is standard, then $S_{{\mathcal R}_+(H),\Omega} = \Gamma_+(S_H)$, \; $J_{{\mathcal R}_+(H),\Omega} =\Gamma_+(J_H)$,
\; $\Delta_{{\mathcal R}_+(H),\Omega}= \Gamma_+(\Delta_H)$. 
\end{itemize}
\end{proposition}
In particular if $O\mapsto H(O)$ is a Poincar\'e covariant net of standard subspaces satisfying
(SS\ref{isotony})-(SS\ref{SSbw}), then its second quantization  $O\mapsto \A(O)={\mathcal R}_+(H(O))$
is a Poincar\'e covariant net of von Neumann algebras satisfying (HK\ref{isotony})--(HK\ref{bw}) in Section \ref{haag-kastler}.
Furthermore, it follows from Proposition \ref{prop:secquant}, that
$\A^\dual(O) = {\mathcal R}_+(H^\dual(O))$, where $\A^\dual(O) := \A(O')'$ is the dual net of $\A$.

\subsubsection{M\"obius covariant nets without chiral components}\label{nochiral}
We apply the BGL construction to a particular class of representations of $G$.
To be specific, let us take the irreducible positive energy representation $U_0$ of $\mob$
with lowest weight $1$. Let $j_0$ be the map $z\mapsto \bar z$ on $S^1$
and it acts on $\mob$ through the identification of $\mob$ as $\mathrm{SU}(1,1)$, which we call $\alpha_0$.
$U_0$ extends to an (anti-)unitary representation of $\mob \rtimes_{\alpha_0} \ZZ_2$ which we denote again by $U_0$,
namely $U_0$ is unitary on $\mob$ and anti-unitary on $j_0\Mob$.
(see e.g.\! \cite[Section 1.6.2]{Longo08}).
We can define a positive-energy (anti-)unitary representation of $G\rtimes_\alpha \ZZ_2$ by $U(g_\L\times g_\R) := U_0(g_\L)\otimes U_0(g_\R)$
and $U(j) := U(j_0)\otimes U(j_0)$ where, on the real line picture, $j_0: \RR \ni a \mapsto -a$.
Note further that the joint spectrum of $U$ restricted to the translation group $\RR^{1+1}$
has no nontrivial spectral projections corresponding to the sets $\RR_+ \times \{0\}$ or $\{0\}\times \RR_+$.
In the second quantization $\Gamma_+(U)$, these spectral projections contain only $\CC\Omega$.

By the BGL construction, we obtain a M\"obius covariant net $(\A_U, \G_+(U), \Omega)$.
This net does not have chiral components, i.e.\!
$\A_{U,\L}^{\mathrm{max}}(I_\L) := \A(I_\L \times I_\R)\cap \left(\G_+(U)(\iota\times\uMob)\right)'= 
\CC\1 =\A(I_\L \times I_\R)\cap \left(\G_+(U)(\uMob\times\iota)\right)' =: \A_{U,\R}^{\mathrm{max}}(I_\R)$
(see \cite[Definition 2.1]{Rehren00}).
Indeed, if these algebras were nontrivial, the representation $\G_+(U)(\tau\times\tau)$ of $\RR^{1+1}$
would have nontrivial spectral projections (properly larger than $\CC\Omega$)
corresponding to the sets $\RR_+ \times \{0\}$ or $\{0\}\times \RR_+$, which is a contradiction.

A M\"obius covariant net $\A$ is said to have the {\bf split property} if
for each $O \Subset \tilde O$ (namely, $\overline{O} \subset \tilde O$)
there is an intermediate type I factor ${\mathcal R}_{O,\tilde O}$ such that $\A(O) \subset {\mathcal R}_{O, \tilde O} \subset \A(\tilde O)$.
To show that $\A_U$ has the split property,
let us consider the theory on restricted to the timelike line $a_\R = a_\L$,
namely $\RR\supset I\mapsto \A_U(O_I)$ where $O_I$ is the minimal double cone including $I$.
The diagonal action $\G_+(U)(g\times g)$ acts on $\{\A(O_I)\}$ covariantly, where $g \in \mob$
and the one-particle conformal Hamiltonian is $L_0\otimes1 + 1\otimes L_0$, where $L_0$ is the rotation generator in $U_0$.
As we took $U_0$ as an irreducible representation of $\mob$,
the one-particle conformal Hamiltonian satisfies the trace class property:
 \[
  \Tr (e^{-\beta(L_0\otimes1 +1\otimes L_0)})= \Tr(e^{-\beta L_0})^2<\infty,\qquad \text{for any}\; \beta>0
 \]
since  $\Tr(e^{-\beta  L_0})<\infty$.
Now the trace class property is preserved through second quantization \cite[Corollary 7.4.2]{Longo08}
and  it ensures the split property for any inclusion $\A_U(O_I)\subset \A_U(O_{\tilde I})$,
with $I\Subset \tilde I$ \cite[Corollary 6.4]{BDL07}.
For any inclusion of double cones $O\Subset \tilde O$,
we can find $O_1$ such that $O\Subset O_1 \Subset \tilde O$
and $O_1\Subset \tilde O$ is conformally equivalent to some $O_I \Subset O_{\tilde I}$
for which the split property holds, therefore,
$\A_U(O)\subset \A_U(\tilde O)$ with $O\Subset \tilde O\subset\RR^{1+1}$ satisfies the split property.

\subsubsection{Dual net without M\"obius covariance}\label{dual}
As $(\A, \G_+(U), \Omega)$ is M\"obius covariant, and especially Poincar\'e-dilation covariant,
the dual net $(\A_U^\dual, \G_+(U), \Omega)$, where $\A(O) = \A(O')' = \A(W_1)\cap \A(W_2)$,
where $W_1, W_2$ are wedges such that $O = W_1\cap W_2$, remains to be Poincar\'e-dilation covariant,
because the set of wedges is closed under Poincar\'e and dilation transformations.
Here we show that it cannot be extended to a M\"obius covariant net
because $\A_U^\dual(V_+) = \B(\H)$. We identify $U$ with the direct integral of massive representations $U_m$
of $\poincare$, for which the property $H_{U_m}(V_+) = \H_m$ is well known.
This provides an example of a M\"obius covariant net $(\A_U, \G_+(U),\Omega)$
whose dual net $(\A_U^\dual, \G_+(U), \Omega)$ neither is M\"obius covariant nor satisfies the split property.

\paragraph{The massive free field.}
 Let $U_m$ be the scalar representation of the Poincar\'e group $\poincarej$ with mass $m$.
 It has the form
\begin{align*}
 (U_{m}(a,\l)\xi)(p_1) &= e^{ia\cdot p_m}\xi(\l^{-1}(p_1)),\qquad (a,\l)\in \poincare \\
 (U_m(j)\xi)(p_1) &= \overline{\xi(p_1)},
\end{align*}
where $p_m(p_1) = (\omega_m(p_1), p_1) \in \RR^{1+1}$
(in $(p_0,p_1)$-coordinate, not in $(p_\L, p_\R)$-coordinate),
$\omega_m(p_1) = \sqrt{m^2 + p_1^2}$,
$\l(p_1) = -\sinh(\l) \,\omega_m(p_1) + \cosh(\l) \,p_1$
(where we identify $\l \in \RR$ and an element of the Lorentz group)
and $\xi\in \H_m=L^2(\RR, \frac{dp_1}{2\omega_m(p_1)})$.
Let $\{H_m(W)\}$ be the net of standard subspaces for wedges associated to $U_m$ as in Section \ref{standard}.
We define subspaces relatively to double cones by duality as in Equation \eqref{eq:HOdual}.
Actually, the more traditional construction of the free massive field net satisfies Haag duality, c.f.\!  \cite{O},
hence the one-particle local subspaces can be explicitly described as
\begin{align*}
\begin{array}{l}
 H_m(O) = \overline{\{\hat f^+(p_m) \in \H_m: f \in \mathscr{S}(\RR^{1+1}, \RR), \supp f\subset O\}}, \\
 H_m(W) = \overline{\{\hat f^+(p_m) \in \H_m: f \in \mathscr{S}(\RR^{1+1}, \RR), \supp f\subset W\}},
\end{array} \, \hat f^+(p_1) = \frac1{2\pi}\int d^2a\, f(a)e^{-ia\cdot p_m}. 
\end{align*}

We associate the following real subspaces to the forward and the backward light cones $V_\pm$:
 \[
  H_m(V_\pm)=\overline{\sum_{O\subset V_\pm} H_m(O)}.
 \]

The following proposition is partly an adaptation of the arguments in \cite{SW}.
\begin{proposition}\label{prop:full} $H_m(V_\pm)=\H_m$
\end{proposition}
\begin{proof}
We prove the claim for $V_-$. It can be proved analogously for $V_+$.

Let $O\subset V_-$ and take the vectors  $\xi \in H_m(O)$ and $\eta \in H_m(V_-)'$.  Consider the function
 \[
  f(a)=\im\langle\eta, U(a)\xi\rangle.
 \]
It follows that $f$ is real, it vanishes for any $a\in V_-$ and, since $(\Box +m^2)f=0$, $\supp\hat f\subset ( -\Omega_m\cup \Omega_m)$
as a distribution.

 We claim that $f\equiv 0$. Let $d\mu(p)=\hat f(p) \,\delta(p^2-m^2)\, d^2p$ be the measure associated to the Fourier transform of $f$, namely
 \[
  f(a_0)=\int e^{ia_0p_0}d\mu_{a_1}(p_0),\qquad a=(a_0,a_1), p=(p_0,p_1)
 \]
 where 
\begin{align*}
d\mu_{a_1}(p_0) & =\int e^{-ia_1p_1}d\mu(p_0,p_1)\\ &=
\left(e^{-ia_1\sqrt{p_0^2-m^2}}f(p_0,\textstyle{\sqrt{p_0^2-m^2}})+e^{-ia_1\sqrt{p_0^2-m^2}}f(p_0,-\textstyle{\sqrt{p_0^2-m^2}})\right)dp_0.\end{align*}
  Now, note that the support of $a_0\mapsto f(a_0,a_1)$ is contained in $\RR^+-|a_1|$,
  thus the Fourier transform $d\mu_{a_1}(p_0)$ extends to an analytic function on the upper complex half-plane.
  Furthermore, $p_0\mapsto d\mu_{a_1} (p_0)$ is null on the interval $(-m,m)$ and hence for any test function $h$,
  the analytic function $h *d\mu_{a_1} (p_0)$ is  $0$ by the reflection principle.
  Therefore,   $f(\cdot,{a_1})\equiv 0$ for every ${a_1}\in\RR$, and hence $f\equiv 0$.

Now let $\eta$ be in the symplectic complement of any $H_m(O)$:
$\eta \in \left(\overline{\sum_{O\subset M}H_m(O)}\right)'$,
hence in particular it belongs to $H_m(W')' \cap H_m(W)'=  H_m(W) \cap H_m(W')$.
Indeed,
\[
\eta\in\bigcap_{O\subset W'}H_m(O)'= \left(\overline{\sum_{O\subset W'} H_m(O)}\right)'=H_m(W')'=H_m(W),
\]
and similarly for $W'$.

We show that $\eta=0$. Firstly we observe that $\eta\in H_m(W)\cap H_m(W')$ hence $\Delta^{it}_{H_m(W)}\eta=\eta$. Indeed, 
\begin{align*}\eta\in \ker (\1-S_{H_m(W)})\cap\ker (\1-S_{H_m(W')})&\Leftrightarrow\\
\eta\in \ker (\Delta_{H_m(W)}^{1/2}-\Delta_{H_m(W)}^{-1/2})&\Leftrightarrow\\ \Delta_{H_m(W)}\eta=\eta.\end{align*}
In particular, by the Bisognano-Wichmann property, we have that  $U(\Lambda_W(t))\eta=\eta$ for any wedge $W$ and $t\in\RR$.
Since $U_m$ is an irreducible representation of $\poincare$,
boosts does not have proper invariant vectors,
and we conclude that $\eta=0$.

\end{proof}

This proof can be adapted in any Minkowski space $\RR^{1+s}$ with $s\geq 1$.

\paragraph{The product representation as the direct integral of massive representations.} 
We have seen in \cite[Section 5.2]{BT15} that
the representation $U_0$, restricted to the translation-dilation group,
can be realized on $L^2(\RR_+, pdp)$ as follows:
\begin{align*}
 (U_0(\tau(t))\xi)(p) &= e^{itp}\xi(p) \\
 (U_0(\delta(s))\xi)(p) &= e^{-s}\xi(e^{-s} p) \\
 (U_0(j_0)\xi)(p) &= \overline{\xi(p)}.
\end{align*}
and accordingly the product representation $U$ restricted to the Poincar\'e-dilation group on
$L^2(\RR_+, p_\L dp_\L) \otimes L^2(\RR_+, p_\R dp_\R)$ is given by
\begin{align*}
 (U(\tau(t_\L)\times \tau(t_\R))\xi)(p_\L, p_\R) &= e^{i(t_\L p_\L + t_\R p_\R)}\xi(p_\L, p_\R) \\
 (U(\delta(s_\L)\times\delta(s_\R))\xi)(p_\L, p_\R) &= e^{-s_\L - s_\R}\xi(e^{-s_\L} p_\L, e^{-s_\R}p_\R) \\
 (U(j_0\times j_0)\xi)(p_\L, p_\R) &= \overline{\xi(p_\L, p_\R)}.
\end{align*}

With this realization and the correspondence
$2p_\L p_\R = p_0^2 - p_1^2 = m^2, p_1 = \frac{p_\R - p_\L}{\sqrt2}$,
we have the natural identification
\begin{align*}
L^2(\RR_+, p_\L dp_L)\otimes L^2(\RR_+, p_\R dp_\R)
&\cong L^2(\RR_+^2, p_\L p_\R dp_\L dp_\R) \\
&\cong L^2(V_+, \textstyle{\frac{m^3\, dmdp_1}{2\omega_m(p_1)}}) \\ 
&= \int^\oplus_{\RR_+} d\mu(m)\, L^2(\RR, \textstyle{\frac{m^3\,dp_1}{2\omega_m(p_1)}})
\end{align*}
This identification is given by the map $\xi(p_\L, p_\R) \mapsto \xi'(m, p_1) = 2\xi(\frac{\omega_m(p_1) + p_1}{\sqrt2}, \frac{\omega_m(p_1) - p_1}{\sqrt2})$,
where the factor $2$ is needed to make it unitary.
It is straightforward to check that this intertwines the representation $U$ above and
\[
 (U'(a,\l)\xi')(m,p_1)=e^{ia\cdot (\omega_m(p_1),p_1)}\xi'(m,\l^{-1}(p_1))\qquad (a,\l)\in\poincare,
\]
acting on 
 $\H= L^2\left(\RR_+\times\RR, d\mu(m)dp\frac{m^3}{2\omega_m(p_1)}\right) =\int^\oplus_{\RR_+} m^3d\mu(m)\,L^2(\RR, \textstyle{\frac{\,dp_1}{2\omega_m(p_1)}})$,
 where the boost $\l$ corresponds to $\delta(\l)\times\delta(-\l)$. 
Then, $U$ decomposes into the direct integral
\[
U=\int_{\RR^+}^\oplus{m^3\,d\mu(m)}\,U_{m}. 
\]

A generalized free field can act on the second quantization of this Hilbert space, and it is covariant with respect to $\G_+(U)$.
It is well known that the net of von Neumann algebras for generalized free fields
does not always satisfy Haag duality \cite{Landau74}, depending on the measure on the space of $m$:
it is proven that if the measure decays exponentially, Haag duality fails \cite{Landau74}.
We will show that the measure $m^3\, d\mu(m)$ is associated with a M\"obius covariant net while
the dual net cannot be made M\"obius covariant.
Accordingly, we conjecture that the generalized free field corresponding to the measure $m^3 d\mu$ fails to have Haag duality.

\paragraph{The dual net $\A^\dual$ is not M\"obius covariant.}
Here we show that the dual net $\A^\dual$ does not satisfy (HK\ref{rsv}),
therefore, it cannot be made M\"obius covariant.
It turns out that, since $U$ is a direct integral of massive representations,
$\A^\dual$ does not satisfy the split property either.

Let $\{H_U(W)\}$ (respectively $\{H^\dual_U(O)\}$) be the covariant (respectively dual) BGL
net of standard subspaces for wedges (respectively double cones)
associated with the representation $U=\int_{\RR^+}^\oplus{m^3\,d\mu(m)}\,U_{m}$.

We show that $H_U^\dual(O)=\int^\oplus_{\RR^+}m^3\,d\mu(m)H_m(O)$,
where the direct integral is taken with $\H_U$ as a real Hilbert space.
Note also that $H_m(O)$ is a $\mu$-measurable family of (real) subspaces
in the sense of \cite[Section II.1.7]{Dixmier81} (see Appendix \ref{directintegral}),
because one can take a sequence of real test functions supported in $O$
which separate the points in $O$ (see Proposition \ref{prop:dix}).
Now, we have that
\begin{align*}
H^d_U(O)=
&{\bigcap_{W\supset O}H_U(W)}
\stackrel{(\bullet)}= {\bigcap_{W\supset O}\int_{\RR_+}^\oplus m^3\, d\mu(m) H_m(W)}\\
\stackrel{(\star)}=&\int_{\RR_+}^\oplus m^3\, d\mu(m){ \bigcap_{W\supset O} H_m(W)}=\int_{\RR_+}^\oplus m^3\, d\mu(m) H_m(O).
\end{align*}
Since $U$ disintegrates as $\int_{\RR_+}^\oplus m^3d\mu(m)U_m$ it is easy to see that  $H_U(W)= \ker (\1 - U(j_W)\Delta_W^\frac12) \supset
\int_{\RR_+}^\oplus m^3\, d\mu(m) H_m(W)$. The converse inclusion follows since the Lorentz boosts
$U(\Lambda_W(t))=\int_{\RR_+}^\oplus m^3\, d\mu(m)U_m(\Lambda_W(t))$ fix the subspace $\int_{\RR_+}^\oplus m^3\, d\mu(m) H_m(W)$ for any $t\in\RR$. Since the latter subspace is standard and (SS\ref{SSbw}) holds on $H_U$, we conclude $(\bullet)$ by Lemma \ref{inc}. To justify 
$(\star)$, we refer to the Appendix B, Lemma \ref{lem:ABC}(b),
and note that we only need two wedges whose intersection is $O$. 
Therefore, by Lemma \ref{lem:ABC}(c), 
\[
 H_U^\dual(V_+) = \overline{\sum_{O\supset V_+}H_U^\dual(O)}=\int_{\RR_+}^\oplus m^3\, d\mu(m) H_m(V_+).
\]
We have seen in Proposition \ref{prop:full} that $H_m(V_+) = \H_m$,
hence it follows that $H_U^\dual(V_+) = \H$.
In particular, (HK\ref{rsv}) fails and the net $\A^\dual$ cannot be made M\"obius covariant
by replacing $\G_+(U)$ by any other representation.

It is known that the dual net of the generalized free field does not satisfy the split property
if the measure is not atomic \cite[Theorem 10.2]{DL84}. Our measure is $m^3\,d\mu(m)$, therefore,
the dual net $\A^\dual$ fails to have the split property\footnote{This can be seen at the
one-particle level. The second quantized net $\A_U$ has the split property if and only if 
the operator
\[
 \int_{0}^{+\infty}\Delta_{F_m^{O,\tilde O}}|_{[0,1]} d\mu(m)
\]
is trace class, where $F_m^{O, \tilde O}$ is the intermediate type I factor subspace  between $H_m(O)\subset H_m(\tilde O)$, cf. \cite{FG94,DL84}. In particular it is a necessary condition for the split property that $d\mu$ has to be purely atomic, concentrated on isolated points, cf. \cite{Morinelli18}.
}.

\section{Comments on assumptions and examples}\label{comments}

\subsection{Operator-algebraic assumptions}\label{oa}
As we discuss the question whether dilation covariance can be promoted to M\"obius covariance,
(HK\ref{dilation}) is a natural assumption.
The physical meaning of (HK\ref{rsv}) and (HK\ref{bwcone}) are not very clear,
but they are necessary conditions for a Haag-Kastler net to be M\"obius covariant.
Indeed, a M\"obius covariant net extends to the cylinder $\cyl$,
any double cone is conformally equivalent to a lightcone, hence (HK\ref{rsv}),
and the Bisognano Wichmann property holds automatically for wedges, double cones and lightcones,
hence (HK\ref{bwcone}).
We showed in Section \ref{generalized} that (HK\ref{rsv}) excludes the dual net of certain generalized free fields
which are counterexamples to the implication ``dilation $\Longrightarrow$ M\"obius''
without any additional assumption.
Furthermore, even if (HK\ref{rsv}) is satisfied,
(HK\ref{bwcone}) may fail and in this case M\"obius covariance cannot be expected,
as shown in Section \ref{internal}.
Conversely, if $(\A, U, \Omega)$ satisfies (HK\ref{isotony})--(HK\ref{bw}) and (HK\ref{rsv}),
the modular group of $\A(V_+)$ commutes with Lorentz boosts $U(\delta(t)\times\delta(-t))$ (because
it preserves $\A(V_+)$ and $\Omega$, hence Theorem \ref{thm:bor} applies),
however, it is unclear whether it acts covariantly on $\A$.

(HK\ref{hsmib}) is certainly a necessary condition for M\"obius covariance, but might seem too strong,
because it requires that $\Delta_{B_\L}^{it}$ acts as a certain M\"obius transformation up to
an inner symmetry. On the other hand, we proved the implication (HK\ref{sa}) $\Longrightarrow$ (HK\ref{hsmib})
(under (HK\ref{isotony})--(HK\ref{bwcone})) and (HK\ref{sa}) does not refer to any M\"obius transformation,
therefore, Corollary \ref{cr:mob} is rather satisfactory.
Furthermore, let us stress that our proofs do not rely on either stress-energy tensor or current,
c.f.\! Section \ref{physics}.
Indeed, there are M\"obius covariant nets without stress-energy tensor.
Therefore, at least from the mathematical point of view, the modular theory is more essential
for M\"obius covariance.

Yet, there are examples of M\"obius covariant net not satisfying (HK\ref{sa}).
The simplest examples are two-dimensional nets $\vir_c\otimes\vir_c$,
where $\vir_c$ is the Virasoro net with central charge $c > 1$.
As the chiral components $\vir_c$ does not satisfy strong additivity \cite[Section 4, P.122]{BS90},
the two-dimensional net fails to satisfy (HK\ref{sa}).
Another family of examples is the derivatives of the $\mathrm{U}(1)$-current \cite[Corollary 2.11]{GLW98},
and one can again construct two-dimensional nets by tensor product which do not satisfy (HK\ref{sa}).
These examples satisfy (HK\ref{hsmib}), hence are covered by Theorem \ref{th:mob}
(although M\"obius covariance for these examples is trivial because each chiral component if $\mob$-covariant).
Interestingly, in both examples the current is missing, and
there is even no stress-energy tensor in the latter case \cite[Proposition 3]{Koester03},
hence it is not covered by the physics arguments (see below).

In the $d=1$ case, Guido, Longo and Wiesbrock have shown that any local net which is
covariant with respect to the translation-dilation group and satisfies the Bisognano-Wichmann property
extends to a M\"obius covariant net on $S^1$ \cite[Theorem 1.4]{GLW98}.
This does not straightforwardly generalize to $d=1+1$ because
the inclusion $\A(W_\R + (0,a_\R))\subset \A(W_\R)$ is in general not a standard inclusion,
and hence, even if one assumes the Bisognano-Wichmann property for wedges, it is not enough
to construct a representation of the group $\uMob\times\uMob$.
This is why we needed the assumption (HK\ref{hsmib}).

\subsection{Physics literature}\label{physics}
The arguments by Zamolodchikov and Polchinski \cite{Zomolodchikov86}\cite{Polchinski88}
are considered a proof in physics literature (e.g.\! \cite{Nakayama15}).
They are expressed in terms of Wightman-type assumptions on pointlike observables $O_k(x)$, i.e.
\cite[Section 3.1]{Nakayama15}:
\begin{itemize}
 \item unitarity
 \item Poincar\'e covariance
 \item (unbroken) dilation covariance
 \item discrete spectrum in scaling dimension:
 each of the pointlike observables $O_k(x)$ has a definite ``scaling dimension''
 $\Ad U(\delta(t))(O_k(x)) = t^{\Delta_k}O_k(tx)$.
 \item existence of scale current: there are a stress-energy tensor $T_{\mu\nu}(x)$
 and a current $J_\mu(x)$ such that $\int d^{d-1}x\,[x^\rho T_{\rho 0}(x) - J_0(x)]$
 is the generator of dilations.
\end{itemize}
Note that, apart from implicit assumptions such as $\{O_k(x)\}$
are Wightman fields
and the existence of the integral for the generator of dilations, the well-definedness of scaling dimension is also a quite strong
assumption: it implicitly says that there is a family of
fields $\{O_k(x)\}$ which {\it linearly} generate the whole Hilbert space from the vacuum.
This is different from the usual Wightman situation where the whole Hilbert space
is generated {\it algebraically}, namely, one is allowed to use polynomial of such (smeared) fields.
Once a complete set of fields $\{O_k(x)\}$ is obtained,
one defines the ``scaling dimension matrix'' $\gamma_{k\ell}$
which is defined by $[D, O_k(x)] = -i(\sum_\ell \gamma_{k\ell} O_\ell(x) + x^\mu\partial_\mu O_k(x))$.
Actually it appears that it is not always diagonalizable only from dilation covariance \cite[Section 2.2, below (2.6)]{Nakayama15}.
Therefore, the discreteness of scaling dimension,
and invertibility of $\gamma_{k\ell}$, strongly anticipates the M\"obius covariance.

Under these assumptions, the main concern of \cite{Polchinski88} is whether one can take a stress-energy tensor
$T_{\mu\nu}$ which has the ``canonical scaling dimension'', namely
$\Ad U(\delta(t)) T_{\mu\nu}(x) = t^2 T_{\mu\nu}(tx)$.
Such a new stress-energy tensor is obtained by the inverse of $\gamma_{k\ell}$.
Once $T_{\mu\nu}$ acquires the canonical scaling dimension,
the L\"uscher-Mack theorem \cite[Theorem 3.1]{FST89} gives a clear explanation
of why it satisfies the Lie algebra of vector fields.

By smearing the stress-energy tensor with a slightly singular function
(whose existence is implicitly assumed, see below), one obtains the generators of the Lie algebra of $\uMob$.
However, we are unable to find arguments or explanation on the following points:
\begin{itemize}
 \item Once a stress-energy tensor with canonical dimension is found, how can one show that
 the rest of the fields $\{O_k(x)\}$ are M\"obius covariant?
 Indeed, if one of $\{O_k(x)\}$ are ``descendant'', one is forced to introduce
 the primary field \cite{Koester03, ENR11}.
 Namely, there is no guarantee that the original set of fields $\{O_k(x)\}$ is sufficient.

 It is not always possible to find such an extended family solely from the representation of $\uMob$.
 Indeed, we have the example from Section \ref{generalized}, where the representation $U$ extends
 to $\uMob\times\uMob$, but the observables do not extend to $\cyl$.
 \item Does the stress-energy really give the generator of rotation?
 Here we have two examples in which stress-energy tensor is not directly connected with rotation.
 \begin{itemize}
  \item The $\mathrm{U}(1)$-current net. One can take a new stress-energy tensor
 $T^c_{\mu\nu}$ with central charge $c > 1$ which are relatively local to the $\mathrm{U}(1)$-current $J$
 and generate the same translation-dilation group, but does not give the correct rotation \cite{BS90}.
  \item The dual net of the Virasoro net $\vir_c$ with $c>1$.
  The stress-energy tensor $T^c_{\mu\nu}$ generate the correct translation-dilation group,
  but not the rotations for the dual net.  
 \end{itemize}
 Besides, in order to obtain the generator of rotation, one has to smear the stress-energy tensor
 with the function $x^2+1$ (c.f\! \cite[below (2.14)]{Nakayama15}), which
 is more singular than other generators ($1$ for translations and $x$ for dilations).
 Indeed, this is exactly why $T^c$ fails to extend to $S^1$ on the whole
 Hilbert space of the $\mathrm{U}(1)$-current.
 It appears to be known that in this case the conformal covariance cannot be obtained \cite{Nakayama18}.
 \item Stress-energy tensor is not uniquely determined, even if it is assumed to have the canonical dimension,
 as in the $\mathrm{U}(1)$-current net. Which one is the ``physically'' correct stress-energy tensor?
 \end{itemize}

 Generalized free fields (c.f.\! Section \ref{generalized}) are referred to as ``fake counterexamples'' in \cite[Section 3.1]{DR09},
 because they do not possess stress-energy tensor.
 In fact, one can construct a stress-energy tensor
 which generate the Poincar\'e symmetry, but it turns out very singular \cite[Section 3]{DR03}.
 Yet, the smeared stress-energy tensor gives a quadratic form on the Wightman domain,
 which means that the Wightman-type assumptions are crucial.
 Besides, the existence of stress-energy tensor in this sense does not depend on the scaling dimension $\Delta$,
 while the extension of $U$ to the M\"obius group (or the conformal group in the case of $(3+1)$-dimensions) depends on $\Delta$.
 Therefore, one might conclude that the absence of conformal covariance
 in generalized free fields for certain $\Delta$ is due to the representation theory of the M\"obius/conformal group,
 and not to the absence of stress-energy tensor.

\subsection{Open problems}

We do not know whether (HK\ref{hsmib}) or (HK\ref{sa}) can be dropped or weakened:
(HK\ref{hsmib}) is surely a necessary condition for M\"obius covariance,
but we do not have an example where (HK\ref{isotony})--(HK\ref{bwcone}) hold
but (HK\ref{hsmib}) fails.
A natural candidate for a counterexample is the generalized free field
with the measure $m^{s}d\mu$ where $s < 1$. It is clear that the {\it field} cannot be $\uMob\times\uMob$
covariant because such a field should possess negative scaling dimension, which violates unitarity.
However, to show that the {\it net} generated by it cannot be extended to a M\"obius covariant net,
one has to show that there is no extension of $U$ which makes the net covariant. The latter is
more difficult than the (non-)covariance of the field, because the possible extension is determined
by the modular group, which is difficult to compute if it is not {\it a priori} M\"obius covariant.

One would naturally expect that a similar result should hold in $d=3+1$.
However, this problem is widely open.
Differently from $d=1+1$ where the group of symmetry is a product $\uMob\times\uMob$,
the conformal group in four dimensions is locally isomorphic to $\mathrm{SU}(1,1)$,
which is a simple Lie group \cite{Mack77}.
In order to extend the representation $U$ of the Poincar\'e-dilation group,
one might use the modular groups of double cones, but they must act in a compatible way.
To us it is unclear how this can be obtained from Poincar\'e-dilation covariance,
even with some additivity property or any other natural condition.

Another important problem raised by the examples in Section \ref{generalized}
is which nice properties are expected to be inherited by the dual net,
under which conditions.
As we have shown that the dual net of a M\"obius covariant net in Section \ref{generalized} fails to have the split property,
it is a natural question whether the dual nets of $\vir_c$ have the split property
(if not, they could not be conformally covariant \cite{MTW17}).

\appendix

\section{The two-dimensional conformal spin-statistics theorem}\label{spinstatistics}
We saw in the proof of Theorem \ref{th:mob} that for
a Haag-Kastler net $(\A, U, \Omega)$ satisfying (HK\ref{isotony})--(HK\ref{hsmib})
$U$ extends to $\uMob\times\uMob$ and the net $\A$ remains locally covariant with respect to $U$.
To conclude it, we have to show that $U$ factors through $G$, where $G$ is the quotient of
$\uMob\times\uMob$ by the normal subgroup generated by $\{\rho_{-2\pi}\times\rho_{2\pi}\}$
(see Section \ref{2dminkowski}).
A similar statement for nets on $\RR$ is known as the conformal spin-statistic theorem
\cite{BGL93, GL96}, and the two-dimensional version is known to experts (see e.g.\! \cite[Proposition 2.1]{KL04-2}),
but proof is missing in literature. For the reader's convenience, we present a self-contained proof.

Let $j$ denote the spacetime reflection, namely $j(a_\L,a_\R)=(-a_\L,-a_\R)$.
We denote by  $j_W:=gjg^{-1}$ the spacetime reflection with respect to the wedge $W$, where $W=gW_R$, $g\in\Poi$.
\begin{lemma}\label{lemma:covariancej}
Let $(\A,U,\Omega)$ a Poincar\'e covariant net of von Neumann algebras on wedges satisfying (HK\ref{isotony})--(HK\ref{bw})
of Section \ref{haag-kastler}. Then $U$ extends to the group $\poincare \rtimes \ZZ_2$ by
the modular conjugation:
\[
 U(j):= J_{\A(W_R),\Omega}=J_{\A(W_L),\Omega}
\]
and wedge algebras are covariant with respect to this extension, namely,
$\Ad U(j)\A(W) = \A(jW)$.
\end{lemma}
\begin{proof}
From the Bisognano-Wichmann property (HK\ref{bw}), wedge duality $\A(W_\R)=\A(W_\L)'$ follows (see e.g.\! \cite[Proposition A.2]{Tanimoto12-2}).
Locality (HK\ref{locality}) and wedge duality properties together
imply that $J_{\A(W_\R),\Omega}\A(W_\R)J_{\A(W_\R),\Omega}=\A(W_\R)'=\A(W_\L).$
Again by (HK\ref{bw}),
$J_{\A(W_\R)}$ commutes with the boosts and by Theorem \ref{thm:bor}
it satisfies the right commutation relations with translations.
For a general wedge $W = \tau(a)W_\R, a\in\RR^{1+1}$,
we have
\begin{align*}
\Ad J_{\A(W_\R)} (\A(W)) &= \Ad \left[J_{\A(W_\R)} U(\tau(a))\right](\A(W_\R))
= \Ad \left[U(\tau(-a))J_{\A(W_\R)}\right](\A(W_\R)) \\
&= \Ad U(\tau(-a))\A(W_\L) = \A(\tau(-a)W_\L), 
\end{align*}
which is the covariance of wedge algebras.
\end{proof}
\begin{remark}
We do not know whether covariance holds for the whole net $\A$
including double cones. It does if we assume Haag duality for double cones,
but that assumption might be too strong for M\"obius covariant nets,
as we saw in Section \ref{dual} that going to the dual net may break M\"obius covariance. 
\end{remark}

Consider the local action of $\uMob\times\uMob$ on the Minkowski space
given by $(g_\L\times g_\R)\cdot(a_\L,a_\R)=(g_\L a_\L, g_\R a_\R)$.
We identify $\RR$ with the universal covering of $S^1$,
and with this identification, $\uMob\times\uMob$ acts on $\RR^2$.
The Minkowski space can be identified with $(-\pi,\pi)\times (-\pi,\pi) \subset \RR \otimes \RR$ and
we denote it by $M_0$.
See \cite{BGL93}.
Any diamond with center in $(a,b)\in\RR^2$
of the form $(a-\pi,a+\pi)\times (b-\pi,b+\pi)$ is a copy of the Minkowski spacetime
and is denoted by $M_{(a,b)}$.
Let  $\theta\mapsto \rho_\theta \in\uMob$ and $t\mapsto \tau_t\in\uMob$
be the lifts of the rotation and translation groups, respectively.
Copies $\{M_{(a,b)}\}$ of the Minkowski spacetimes are transformed to each other
by some element in $\RR\times\RR\ni(\theta_1,\theta_2)\mapsto \rho_{\theta_1}\times \rho_{\theta_2}\in\uMob\times\uMob$.
Let $t\mapsto \Lambda_W(t)$  be the lift of the boosts associated with the wedge $W$.
Wedges and double cones are diamonds with sides shorter than $2\pi$, and in $\RR^2$ they are indistinguishable.
We call then double cones (in $\RR^2$).

\begin{figure}[ht]\centering
\begin{tikzpicture}[scale=0.75]
        \draw [->] (-3,0) --(5,0) node [above right] {$\frac{a_\R-a_\L}{\sqrt{2}}$};
         \draw [->] (0,-3)--(0,4) node [ right] {$\frac{a_\R+a_\L}{\sqrt 2}$};
          \draw [thick] (-2,0)-- (0,2) node [ left] {$M_0$};
          \draw [thick] (0,2)-- (2,0);
          \draw [thick] (-2,0)-- (0,-2);
          \draw [thick] (2,0)-- (0,-2);
          
           \draw [ thick] (0,0)-- (1,1);
          \draw [thick] (0,0)-- (1,-1);
          
\draw [ thick] (-0.4,0) -- (-0.9,0.5);
\draw [thick] (-0.4,0) -- (-0.9,-0.5);
\draw [thick] (-1.4,0) -- (-0.9,0.5);
\draw [thick] (-1.4,0) -- (-0.9,-0.5);
          
          \draw [thick,dotted] (-2,-3)--(-2,4);
          \draw [thick,dotted] (2,-3)--(2,4);
          
          \draw [thin,->] (-3,-3) -- (3.5,3.5) node [above right] {$a_R$};
          \draw [thin,->] (3,-3) -- (-3.5,3.5) node [above left] {$a_L$};
  \fill [color=black,opacity=0.2]
               (-0.4,0) -- (-0.9,0.5) -- (-1.4,0) --(-0.9,-0.5);
           \fill [color=black,opacity=0.6]
              (0,0) -- (1,1) -- (2,0) -- (1,-1);

\end{tikzpicture}
\caption{The Minkowski space $M_0$ in $\RR^2$. The cylinder is obtained by identifying the dotted lines.
The dark grey region is the right wedge $W_\R\subset M_0$ and the light grey region is a doublecone $O\subset M_0$}
\label{fig:1}
\end{figure}
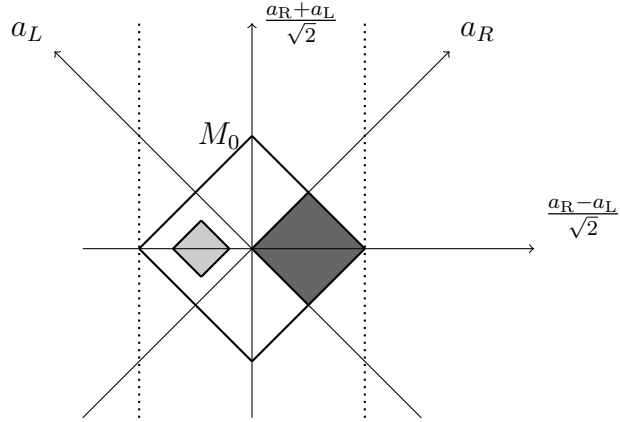

Let $(\A,U,\Omega)$ satisfy (HK\ref{isotony})--(HK\ref{bw}) and
assume the local $\uMob\times\uMob$-covariance:
\begin{itemize}
\item[(LM)] $U$ extends to $\uMob\times \uMob$ and $U(g) \A(O)U(g)^* =\A(gO),$ for $g$ in a small
neighborhood $\U_O$ of $\i\times\i\in\uMob\times \uMob$
such that $\U_O O$ stays in $M_0$.
\end{itemize}
The net $\A$ of von Neumann algebras can be extended to any double cone in $\RR^2$ (diamonds with sides shorter than $2\pi$)
by covariance. Isotony (HK\ref{isotony}) and locality (HK\ref{locality}) continue to hold
for double cones included in one copy $M$ of Minkowski space.
For a double cone which is a wedge $W$ in $\M_0$, there are a priori two definitions: $\A(W) = \bigvee_{D\subset W}\A(D)$
and by covariance. But actually they coincide by the continuity of $U$,
hence we only have to deal with $\{\A(D)\}$ where $D$ is a ``double cone'' in $\RR$.

If $\Ad U(\rho_{-2\pi}\times \rho_{2\pi})\A(O) = \A(O)$ holds for any double cone in $\RR^2$,
we can identify points connected by $\rho_{-2\pi}\times \rho_{2\pi}$ and
obtain a net on $\cyl$. In this case, we say that $(\A, U, \Omega)$ reduces to the cylinder $\cyl$.

\begin{lemma} \label{lemma:cyl} 
Let $(\A,U,\Omega)$ be a $\uMob$ covariant net satisfying (HK\ref{isotony})--(HK\ref{bw}) and (LM).
Then, it reduces to $\cyl$.
\end{lemma}
\begin{proof}
We show that for any double cone $O \subset M_0$
it holds that $\Ad U(\rho_{-2\pi}\times \rho_{2\pi})\A(O) = \A(O)$.
As any double cone on $\RR^2$ are contained in a copy $M$ of the Minkowski space,
and the net $\{\A(O)\}$ is defined by covariance with respect to $\rho\times\rho$,
this suffices to conclude
that $\Ad U(\rho_{-2\pi}\times \rho_{2\pi})\A(O) = \A(O)$ holds for any double cone.

Let $W_{\R / \L}$ and $W^\pi_{\R / \L}$  be the right and left wedges in $M_0$
and $M_{(-\pi,\pi)}=(\rho_{-\pi}\times \rho_{\pi})M_0$, respectively.
Note that $W_\R$ coincides with the left wedge $W^\pi_\L$ in $M_{(-\pi,\pi)}$, see Figure \ref{fig:wed}.
Therefore, by wedge duality, $\A(W_\R)'=\A(W_\L)$ and $\A(W_\R)'=\A(W^\pi_\L)'=\A(W_\R^\pi)$. Altogether,
\begin{equation}\label{eq:period}
A(W_\L) = \Ad U(\rho_{-2\pi}\times \rho_{2\pi})(\A(W_\L))=\A(W^\pi_\R).
\end{equation}

\begin{figure}[ht]\centering 
\begin{tikzpicture}[scale=0.75]
         \draw [->] (-3,0) --(5,0);
         \draw [->] (0,-3)--(0,4); 
         \draw [thick] (-2,0)-- (0,2) node [ left] {$M_0$};
          \draw [thick] (0,2)-- (2,0);
          \draw [thick] (-2,0)-- (0,-2);
          \draw [thick] (2,0)-- (0,-2);
          
          \draw [thick] (0,0)-- (2,2)  node [right] {$M_\pi$};
          \draw [thick] (2,2)-- (4,0);
          \draw [thick] (0,0)-- (2,-2);
          \draw [thick] (4,0)-- (2,-2);
          
          \draw [thick] (0,0)-- (-1,1);
          \draw [thick] (0,0)-- (-1,-1);

          \draw [thick] (2,0)-- (3,1);
          \draw [thick] (2,0)-- (3,-1);
          
                \draw [thick,dotted] (-2,-3)--(-2,4);
          \draw [thick,dotted] (2,-3)--(2,4);
          \fill [color=black,opacity=0.2]
               (-2,0) -- (-1,1) -- (0,0) -- (-1,-1);
           \fill [color=black,opacity=0.2]
               (4,0) -- (3,1) -- (2,0) -- (3,-1);
           \fill [color=black,opacity=0.6]
               (0,0) -- (1,1) -- (2,0) -- (1,-1);
\end{tikzpicture}
\caption{Light grey areas are wedge regions in the two copies of the Minkowski spacetimes $M_\textbf{0}$ and $M_{\pi}$.
Equation \eqref{eq:period} shows that algebras associated with these regions are the equal.}
\label{fig:wed}
\end{figure}
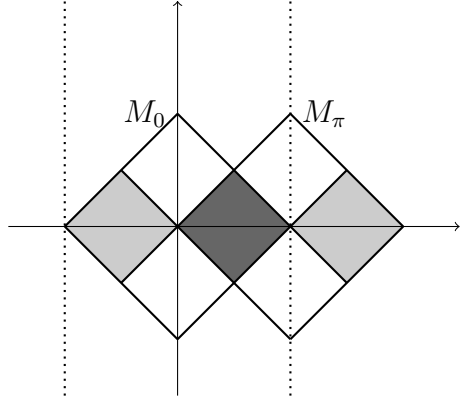

By composing with an appropriate $\tau_t\times \tau_s$, \eqref{eq:period} holds for any wedge $W$ in $M_0$.
It also holds for double cones. Indeed, $W_\R$ is a double cone in $M_{(-\epsilon,\epsilon)}$ with $\epsilon>0$
(see Figure \ref{fig:epsilon})
and by covariance one can infer that for any double cone $O\subset M_0$ 
\[
\A(O)=\Ad U(\rho_{-2\pi}\times \rho_{2\pi})(\A(O)) = \A((\rho_{-2\pi}\times \rho_{2\pi})O).
\]

\begin{figure}[ht]\centering
\begin{tikzpicture}[scale=0.75]
        \draw [->] (-3,0) --(5,0);
         \draw [->] (0,-3)--(0,4);
          \draw [thick] (-2,0)-- (0,2) node  [left] {$M_0$};
          \draw [thick] (0,2)-- (2,0);
          \draw [thick] (-2,0)-- (0,-2);
          \draw [thick] (2,0)-- (0,-2);

          \draw [thick] (-1,0)-- (1,2) node [right] {$M_{(-\epsilon,\epsilon)}$};
          \draw [thick] (1,2)-- (3,0);
          \draw [thick] (-1,0)-- (1,-2);
          \draw [thick] (3,0)-- (1,-2);

           \draw [thick] (0,0)-- (1,1);
          \draw [thick] (0,0)-- (1,-1);          
          
          
                   \fill [color=black,opacity=0.6]
               (0,0) -- (1,1) -- (2,0) -- (1,-1);
\end{tikzpicture}
\caption{The shaded region represents the right wedge in $M_0$ and a double cone in $M_{(-\epsilon,\epsilon)}$}
\label{fig:epsilon}
\end{figure}
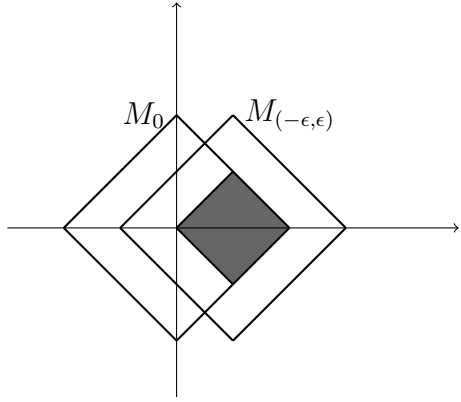

\end{proof}

\begin{lemma}\label{lm:bwj}
Let $(\A,U,\Omega)$ be a $\uMob$ covariant net satisfying (HK\ref{isotony})--(HK\ref{bw})
and (LM).
Then, for any wedge $W$ in a copy of Minkowski space $M$, the wedge modular conjugation $J_W$ implements $j_W$ on $M$.
\end{lemma}
\begin{proof}
The net can be extended to the cylinder by Lemma \ref{lemma:cyl}.
Let $W_\L=(0,\pi)\times(-\pi,0)$ and $W_{\L,1}=(1,\pi)\times (-\pi,-1) =(\tau(1)\times\tau(-1))W_\L\subset W_\L$
be wedges in $M_0$. The restriction of the net $\A$ on $M_0$
satisfies the assumptions of Lemma \ref{lemma:covariancej}, hence we have
 \begin{equation}
 \label{eq:dcdc}
 \Ad U(j)(\A(W_{\L,1})) = \A (W_{\R,1}),
 \end{equation}
 where $W_{\R,1}=(-\pi,-1)\times (1,\pi)$. Now, note that $U(j)=J_{W_\L}=J_{W_\R}$ and that on $M_{(-\pi,\pi)}$
 the region $W_{\R,1}$ is a double cone (see Figure \ref{fig:wedgej}). Thus, by \eqref{eq:dcdc},
\[
 \Ad U(j)(\A(W_{\R,1})) = \A(W_{\L,1}) = \A((\rho_{-2\pi}\times\rho_{2\pi})W_{\L,1}),
\]
where the last equality follows because the net is defined on $\cyl$.
Note that $(\rho_{-2\pi}\times\rho_{2\pi})W_{\L,1}$ is the double cone
which is obtained by reflecting $W_{\R,1}$ by $j_{W^\pi_\R}$ in $M_{(-\pi, \pi)}$,
namely, this is covariance of double cone algebras in $M_{(-\pi, \pi)}$.
We obtain covariance for any other double cone by $\poincare$-covariance, and
this result can be brought back to $M_0$ by $\Ad U(\rho_{-2\pi}\times \rho_{2\pi})$.

\begin{figure}[ht]\centering
\begin{tikzpicture}[scale=0.75]
         \draw [->] (-3,0) --(5,0);
         \draw [->] (0,-3)--(0,4);
          \draw [thick,dotted] (-2,-3)--(-2,4);
          \draw [thick,dotted] (2,-3)--(2,4);

          \draw [thick] (-2,0)-- (0,2) node [ left] {$M_0$};
          \draw [thick] (0,2)-- (2,0);
          \draw [thick] (-2,0)-- (0,-2);
          \draw [thick] (2,0)-- (0,-2);
          
          \draw [thick] (0,0)-- (2,2)  node [right] {$M_\pi$};
          \draw [thick] (2,2)-- (4,0);
          \draw [thick] (0,0)-- (2,-2);
          \draw [thick] (4,0)-- (2,-2);
          
          \draw [thick] (0,0)-- (-1,1);
          \draw [thick] (0,0)-- (-1,-1);

           \draw[thick](1,0)--(1.5,0.5);
          \draw[thick](1,0)--(1.5,-0.5);
          
          \draw[thick](3,0)-- (2.5,0.5);
          \draw[thick](3,0)-- (2.5,-0.5);

          \draw[thick](-1,0)--(-1.5,-0.5);
          \draw[thick](-1,0)--(-1.5,0.5);

          \draw [thick] (2,0)-- (3,1);
          \draw [thick] (2,0)-- (3,-1);
          
          \fill [color=black,opacity=0.2]
               (-2,0) -- (-1.5,-0.5) -- (-1,0) -- (-1.5,0.5);
           \fill [color=black,opacity=0.2]
               (2,0) -- (2.5,0.52) -- (3,0) -- (2.5,-0.5);
           \fill [color=black,opacity=0.6]
               (1,0) -- (1.5,0.5) -- (2,0) -- (1.5,-0.5);

\end{tikzpicture}
\caption{$U(j)$ transforms the algebras in light grey to dark grey regions (and vice versa). They are wedge and doublecone algebras in the $M_0$ and $M_\pi$ pictures, respectively.}
\label{fig:wedgej}
\end{figure}
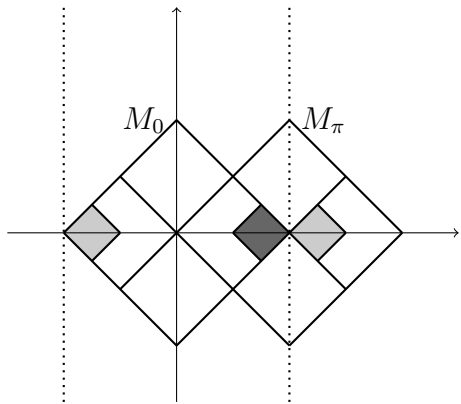
\end{proof}

 Now we can prove the conformal spin-statistics theorem.
\begin{theorem}\label{th:ss}
Let $(\A,U,\Omega)$ satisfy (HK\ref{isotony})--(HK\ref{bw}) and (LM).
Then the representation $U$ of $\uMob \times \uMob$ factors through a
representation of $G$.
\end{theorem}

\begin{proof} We have to show that $U(\rho_{-2\pi}\times \rho_{2\pi}) = \1$.
Let  $O_1 = (-\frac\pi 2,\frac\pi 2)\times(-\frac\pi 2, \frac\pi 2)$ be a double cone in $M_0 = (-\pi,\pi)\times(-\pi,\pi)$.
We saw in Lemma \ref{lm:bwj} that $\Ad J_{\A(W_R),\Omega}$ implements the reflection $j$,
hence it is an antilinear automorphism of $\A(O_1)$. In particular, $J_{\A(W_R),\Omega}$ commutes with $J_{\A(O_1),\Omega}$
and the unitary $J_{\A(W_R),\Omega}J_{\A(O_1),\Omega}$ is self-adjoint.

We are going to show that
\begin{align}\label{eq:jju}
 J_{\A(W_R),\Omega}J_{\A(O_1),\Omega}=U(\rho_{-\pi}\times \rho_{\pi}).
\end{align}
Let $L_{0, \L}$ be the generator of $\theta\mapsto U(\rho_\theta\times \i)$, 
$P_\L$ be the  generator of $t\mapsto U(\tau_t\times\i)$ and 
$\bar P_\L$ be the generator of anti-translations $U((\rho_\pi\circ \tau_t\circ\rho_{-\pi})\times \i)$.
Analogously we introduce $L_{0,\R}$, $P_{\R}$ and $\bar P_{\R}$.
As operators on an appropriate domain, 
\[
 L_{0,\R}=\frac{P_{\R}+\bar P_{\R}}2,\qquad  L_{0,\L}= \frac{P_{\L}+\bar P_{\L}}2,
\]
see e.g.\! \cite[Proposition 1]{Koester02}. By applying Theorem \ref{thm:bor}
to $\A(W_\R)$ and translations or anti-translations,
$J_{\A(W_R),\Omega}$ commutes with each of $P_{\R}, \bar P_{\R}, P_{\L}, \bar P_{\L}$,
thus by antilinearity of $J_{\A(W_R)}$,
\begin{align*}
 J_{\A(W_R),\Omega}U(\rho_{\theta_1}\times \rho_{\theta_2})J_{\A(W_R),\Omega}
 &=J_{\A(W_R),\Omega}\cdot e^{i\theta_1 L_{0,\L} + i\theta_2 L_{0,\R}}\cdot J_{\A(W_R),\Omega} \\
 &=e^{-i\theta_1 L_{0,\L} - i\theta_2 L_{0,\R}} = U(\rho_{-\theta_1}\times \rho_{-\theta_2}).
\end{align*}
The claimed equation \eqref{eq:jju} follows since $\A(W_R) = \Ad U(\rho_{-\frac\pi 2}\times \rho_{\frac \pi 2}){\A(O_1)}$,
therefore,
\[
 J_{\A(W_R),\Omega} = U(\rho_{-\frac\pi 2}\times \rho_{\frac \pi 2})J_{\A(O_1),\Omega}U(\rho_{-\frac\pi 2}\times \rho_{\frac \pi 2})^*.
\]
From Equation \eqref{eq:jju} we conclude that
$U(\rho_{-2\pi}\times \rho_{2\pi}) = (J_{\A(W_R),\Omega}J_{\A(O_1),\Omega})^2 = \1$, hence the Lemma. 
\end{proof}

\section{Basic properties of direct integrals}\label{directintegral}
 Here we follow \cite[Section II.1]{Dixmier81}
 and supply some additional results concerning direct integral of real and complex Hilbert spaces.
 
 Given a field of Hilbert spaces $m\mapsto \H_m$ on a standard measure space $(X,\mu)$
 we can construct the direct integral Hilbert space $\int_X^\oplus \H_m d\mu(m)$ if the field is $\mu$-measurable:
 This definition requires and depends on the choice of a linear subspace $\mathcal{S}$
 of $\Pi_{m\in X}\H_m$ which are by definition $\mu$-measurable vector fields 
 ($\mathcal{S}$ must (i) consist of fields whose norm is measurable,
 (ii) be complete in the sense that it contains any vector field whose pointwise inner product with any other field in $\mathcal{S}$ is measurable
 (iii) and contain a sequence which is total at any point $m \in X$,
 see \cite[Section II.1.3, Definition 1]{Dixmier81}). Note that given a sequence of measurable vector field $\xi_n$
 $\mu$-a.e.\! pointwise converging to $\xi$, namely $\|\xi_n(m)-\xi(m)\|\overset{n\to\infty}\longrightarrow 0$
 for $\mu$-a.e.\! $m\in X$, then $\xi$ is a $\mu$-measurable vector field.
 We also recall that a vector field of bounded operators $m\mapsto T_m\in \B(\H_m)$ is $\mu$-measurable
 if for any $\mu$-measurable field $m\mapsto \xi(m)\in\H_m$ then $m\mapsto T_m\xi(m)\in \H_m$ is $\mu$-measurable. 
 In our concrete case, we take $X = \RR, \H_m = L^2(\RR, \frac{dp}{2\omega_m(p)})$ and $\mathcal S$ consists of Lebesgue-measurable functions in $\RR\times\RR$. 

 Most of the results \cite[Section II.1]{Dixmier81} which we need are written for complex Hilbert spaces,
 but actually one can consider direct integral of real Hilbert spaces and similar results
 hold\footnote{We use caligraphic letters (such as $\H_m$) for complex Hilbert spaces and
 roman letters (such as $H_m$) for real Hilbert spaces.}.
 If $\H = \int^\oplus_X \H_m d\mu(m)$ with the scalar product $\<\cdot,\cdot\>$ is the direct integral
 of complex Hilbert spaces $\H_m$ with the scalar product $\re \<\cdot,\cdot\>_m$,
 $\H$ as a real Hilbert space with $\re \<\cdot,\cdot\>$ can be seen as the direct integral
 of real Hilbert spaces $\H_m$ with $\re\<\cdot, \cdot \>_m$: the $\mu$-measurable set $\mathcal{S}$ can be regarded as the set of
 real $\mu$-measurable vector fields
 ((i) The norm does not change. (ii) If $\xi \in \mathcal{S}$, then also $i\xi \in \mathcal{S}$
 and $\<\eta, \xi\>_m$ is measurable if and only if both of $\re\<\eta, \xi\>_m$ and
 $\re\<\eta, i\xi\>_m = -\im\<\eta, \xi\>_m$ are measurable. (iii) The metric does not change,
 hence from a total sequence $\{\xi_n\}$ with respect to the complex scalar product
 we can make a sequence $\{\xi_n, i\xi_n\}$ which is total with respect to the real scalar product).

 In order to properly describe the BGL-net associated to a direct integral representation in Section \ref{dual}, we 
 need the following two propositions, see \cite{Halperin62} and \cite[Section II.1.7, Proposition 9]{Dixmier81}, respectively.
 The proofs work for $\CC$ as well as for $\RR$.
 
 \begin{proposition}\label{prop:inters}
Let $\K_1,\ldots,\K_k$ be closed subspaces of a Hilbert space $\H$
and $\K=\bigcap_k \K_k$. Let $P_{\K_j}$ and $P_\K$ be the associated orthogonal projections, then
$\left( P_{\K_1},\ldots P_{\K_k}\right)^n$ strongly converge to $P_\K$, as $n\rightarrow+\infty$.
\end{proposition}
\begin{proposition}\label{prop:dix}
 Let $m\mapsto \H_m$ be a $\mu$-measurable field of complex Hilbert spaces over a measure space $(X,\mu)$.
 Let $H_m$ be a closed linear subspace of $\H_m$ and $E_H(m)$ be the  projection onto $H_m$. Let $\mathcal S_H$ be the set of all measurable vector fields
 $m\mapsto \xi(m)$ such that $\xi(m)\in H_m$. Then the following are equivalent:
 \begin{itemize}
\item[(i)] the field of subspaces $m\mapsto H_m$, endowed with $\mathcal S_H$, is $\mu$-measurable;
\item [(ii)] There exists a sequence $\{\xi_n\}_{n\in\NN}$ of $\mu$-measurable vector fields  ($m\mapsto \xi_n(m)\in \H_m$)
such that $\{\xi_n(m)\}_{n\in\NN}$ is a total sequence in each $H_m$;
\item [(iii)] for any measurable vector field $\xi$ with respect to $\mathcal S$, the field $m\mapsto E_H(m)\xi(m)$ is measurable.
\end{itemize}
\end{proposition}

This proposition allows us to consider the direct integral $\int^\oplus_X H_m d\mu(m)$ of standard subspaces.
 A vector $\xi$ in this subspace is a $L^2$-measurable field of vectors $\{\xi(m)\}$
 such that $\xi(m) \in H_m$ for almost all $m$.
 In our application, $H_m$ is the closure of
 the Fourier transforms of real compactly supported continuous functions on the Minkowski space $\RR^{1+1}$,
 restricted to the mass shell with mass $m$.

\begin{lemma}\label{lem:ABC}With the definitions in Lemma \ref{prop:dix},
let $H_m$ be a $\mu$-measurable field of subspaces, and consider $H=\int_X^\oplus H_md\mu(m)$,
\begin{itemize}
\item[(a)]
$m\mapsto H^\perp_m$ $m\mapsto H^\perp_m$ is a $\mu$-measurable field and $H^\perp=\int_X^\oplus H^\perp_md\mu(m)$.
If $\H$ is a real Hilbert space obtained from a complex Hilbert space and $\re\<\cdot,\cdot\>$,
then $m\mapsto H'_m$ $m\mapsto H'_m$ is a $\mu$-measurable field and $H'=\int_X^\oplus H'_md\mu(m)$.

\item[(b)] let $\{H_k\}_{k\in\NN}$ be a countable family of $\mu$-measurable fields, then $\bigcap_{k\in\NN} (H_k)_m$ is a $\mu$-measurable field and  $\bigcap_{k\in\NN} H_k=\int_X^\oplus \bigcap_{k\in\NN} (H_k)_m d\mu(m)$.

\item[(c)] let $\{H_k\}_{k\in\NN}$ be a countable family of $\mu$-measurable fields, then $\overline{\sum_{k\in\NN} (H_k)}_m$ is a $\mu$-measurable field and  $\overline{\sum_{k\in\NN} H_k}=\int_X^\oplus \overline{\sum_{k\in\NN} (H_k) }_md\mu(m)$.
\end{itemize}
\end{lemma}
\begin{proof}
(a) By definition $H^\bot=\{\xi\in\H: \langle \xi,\eta\rangle=0, \eta\in H\}$.
Since $H_m$ is a measurable field, by Proposition \ref{prop:dix} (iii)  for any $\xi\in\H$,
$m\mapsto E_H(m)\xi(m)\in S_H$ is a measurable map.
Then for every $\xi\in\H$, the vector field $m\mapsto (1-E_H(m)) \xi(m)$ is clearly $\mu$-measurable,
hence by implication (iii) $\Rightarrow$ (i) in Proposition \ref{prop:dix},
we have that  $ m\mapsto H_m^\bot$ is a $\mu$-measurable vector field of subspaces.
We denote their direct integral by $\int_X^\oplus d\mu(m) H_m^\bot$.
The inclusion $H^\bot\supset\int_X^\oplus d\mu(m) H_m^\bot$ is obvious.
Any $\xi$ can be decomposed as $E_H\xi + (1-E_H)\xi$,
hence $\int_X^\oplus d\mu(m) H_m^\bot$ must coincide with $H^\perp$.

Now consider $m\mapsto H_m$ $\mu$-measurable field of real subspaces,
then $m\mapsto iH_m$ is also a measurable of real subspaces.
Indeed, let $E_m$ be the projection on $H_m$ for any $x\in \H$, $m\mapsto iE(m)x(m)$ is a measurable vector field,
as $\mathcal{S}$ is closed under the multiplication by $i$.
Therefore, (i) in Proposition \ref{prop:dix} holds
and we obtain $iH = \int^\oplus_X iH_m d\mu(m)$.
We conclude by recalling that $H'=(iH)^\bot$, and combining the last comments, that $H'=\int_X^\oplus d\mu(m)H_m'$.

\vspace{0.3cm}
(b)
 Let $\{H_k\}_{k\in\NN}$ be a family of measurable fields of subspaces. We need to show that $m\mapsto \bigcap_k(H_m)_k$ is a measurable field of real spaces and $\int_X^\oplus d\mu(m)\left(\bigcap_k(H_m)_k\right)\subset\bigcap H_n$. 
Firstly, if the family is finite, namely if we have $H_1, \ldots, H_K$, then for any $\xi\in\H$,
then $((P_{H_1})_m\ldots (P_{H_K})_m)^n\xi(m)\stackrel{n\rightarrow+\infty}\longrightarrow (P_{\cap_k H_k})_m\xi(m)$
for each $m$ by Proposition \ref{prop:inters},
and this is measurable. Hence by implication (iii) $\Rightarrow$ (i) in Proposition \ref{prop:dix}
we conclude that  $m\mapsto \bigcap_{k=1}^K(H_m)_k$ is a measurable field of real spaces.
If the family is countably infinite, we take $\tilde P_K:m\mapsto (\tilde P_K)_m$ as the projection on
$\bigcap_{k=0}^K (H_k)_m$.
This is a $\mu$-measurable family of decreasing projections,
thus for any $\xi\in H$ the  limit $(\tilde P_K)_m\xi(m)\stackrel{K\rightarrow+\infty}\longrightarrow (P_{\cap^\infty_{k=1}H_k})_m\xi(m)$  is still measurable. Thus if $H_k$ is a countable  family of real subspaces, $m\mapsto \bigcap_k (H_k)(m)$ gives a measurable family of real subspaces.
The same sequence of projections shows that, if $\xi \in \bigcap_k H_k$,
then $\xi = P_{\cap_k H_k}\xi$ and hence $\xi(m) \in H_m$, which concludes the claim.

\vspace{0.3cm}

(c) This follows by combining (a) and (b).
\end{proof}

\subsubsection*{Acknowledgements}
We thank Yu Nakayama for interesting discussions and bibliographical information.

The authors acknowledge the MIUR Excellence Department Project awarded
to the Department of Mathematics, University of Rome Tor Vergata, CUP E83C18000100006.

{\small
\def\cprime{$'$} \def\polhk#1{\setbox0=\hbox{#1}{\ooalign{\hidewidth
  \lower1.5ex\hbox{`}\hidewidth\crcr\unhbox0}}} \def\cprime{$'$}

}
\end{document}